\documentclass{article}

\usepackage[preprint]{neurips_2025}
\usepackage{todonotes}

\newcommand{\lrp}[1]{\left({#1}\right)}

\newcommand{\KLinf}[0]{\operatorname{KL_{inf}}}
\newcommand{\Exp}[1]{\mathbb{E}\lrp{#1}}
\newcommand\numberthis{\addtocounter{equation}{1}\tag{\theequation}}

\usepackage[utf8]{inputenc} 
\usepackage[T1]{fontenc}    
\usepackage{hyperref}       
\usepackage{url}            
\usepackage{booktabs}       
\usepackage{amsfonts}       
\usepackage{nicefrac}       
\usepackage{microtype}      
\usepackage{xcolor}         
\usepackage{microtype}
\usepackage{graphicx}
\usepackage{subfigure}
\usepackage{booktabs} 
\usepackage{amsmath}
\usepackage{amssymb}
\usepackage{mathtools}
\usepackage{amsthm}
\usepackage{algorithmic}
\usepackage{algorithm}
\usepackage{array}
\usepackage{amssymb, amsmath}
\usepackage{textcomp}
\usepackage{url}
\usepackage{verbatim}
\usepackage{graphicx}
\usepackage{amsthm}
\usepackage{cite}
\usepackage{multirow}
\usepackage{multicol}
\usepackage{booktabs}
\usepackage{siunitx}
\usepackage{listings}
\usepackage{dsfont}
\usepackage{amssymb}
\usepackage{booktabs, xcolor, colortbl}
\usepackage{caption}
\usepackage[most]{tcolorbox}

\allowdisplaybreaks[1]

\theoremstyle{plain}
\newtheorem{theorem}{Theorem}[section]

\newtheorem{lemma}[theorem]{Lemma}
\newtheorem{corollary}[theorem]{Corollary}
\theoremstyle{definition}
\newtheorem{definition}[theorem]{Definition}
\newtheorem{assumption}[theorem]{Assumption}
\theoremstyle{remark}
\newtheorem{remark}[theorem]{Remark}
\newtheorem{example}[theorem]{Example}

\title{Regret Tail Characterization of Optimal Bandit Algorithms with Generic Rewards}

\author{%
  Subhodip Panda \\
  Department of ECE\\
  Indian Institute of Science\\
  Bangalore, India \\
  \texttt{subhodipp@iisc.ac.in} \\
  \And
  Shubhada Agrawal\\
  Department of ECE\\
  Indian Institute of Science\\
  Bangalore, India \\
  \texttt{shubhada@iisc.ac.in} \\
}

\begin{document}

\maketitle

\begin{abstract}
We study the tail behavior of regret in stochastic multi-armed bandits for algorithms that are asymptotically optimal in expectation. While minimizing expected regret is the classical objective, recent work shows that even such algorithms can exhibit heavy regret tails, incurring large regret with non-negligible probability. Existing sharp characterizations of regret tails are largely restricted to parametric settings, such as single-parameter exponential families. 

In this work, we extend the $\KLinf$-UCB algorithm of \citet{agrawal21heavytailed} to a broad nonparametric class of reward distributions satisfying mild assumptions, and establish its asymptotic optimality in expectation. We then analyze the tail behavior of its regret and derive a novel upper bound on the regret tail probability. As special cases, our results recover regret-tail guarantees for both bounded-support and heavy-tailed (moment-bounded) bandit models. Moreover, for the special case of finitely-supported reward distributions, our upper bound matches the known lower bound exactly. Our results thus provide a unified and tight characterization of regret tails for asymptotically optimal $\operatorname{KL}$-based UCB algorithms, going beyond parametric models.
\end{abstract}

\section{Introduction}
The multi-armed bandit (MAB) problem is a classical framework for sequential decision-making under uncertainty, in which an algorithm interacts with a set of $K$ arms, each associated with an unknown reward distribution. At each round, the algorithm selects an arm and observes a stochastic reward, with the goal of maximizing the total expected reward till time $T>0$. This objective is typically formalized via the notion of regret, defined as the difference between the cumulative reward of an algorithm over a horizon $T$ and that of an oracle policy that always plays the arm with the largest mean reward. Minimizing expected regret is equivalent to maximizing expected cumulative reward, and serves as the standard performance criterion in bandit problems. 

A dominant focus of the contemporary literature has been on designing algorithms that minimize the \emph{expected} regret. \citet{LaiRobbins85,BURNETAS1996} proposed a fundamental lower bound on the expected cumulative regret suffered by any reasonable algorithm for this setup. Algorithms achieving this lower bound (even in the multiplicative constants) have also been developed for a wide class of reward distributions~\citep{HondaBounded10, honda2011asymptotically, SAgrawal2011TS1, SAgrawal2012TS2, kaufmann2012TS, cappe2013kullback, honda2015SemiBounded, agrawal21heavytailed}. We refer the reader to \citet{lattimore2020bandit,bubeck2012Survey} for a comprehensive treatment of this classical setup.

While expected regret is a canonical performance metric in stochastic bandits, it provides only a coarse summary of algorithmic behavior. A more informative characterization is obtained by studying the distribution of the regret~\citep{Audibert2009UCBV, salomon2011deviations, KalvitZeevi2021WorstCaseBandits}, and in particular its tail behavior. Understanding the upper tail of the regret distribution is crucial for quantifying the probability of rare but severe failure events in which an algorithm may incur considerable regret. Such high-regret outcomes can be especially consequential in applications such as clinical trials, where such events may correspond to exposing a large number of patients to suboptimal treatments. Consequently, controlling heavy regret tails, rather than merely minimizing expected regret, is a key objective in risk-sensitive and safety-critical settings. Thus, a theoretical understanding of regret tail probabilities of existing algorithms can provide a principled avenue for identifying and mitigating the vulnerabilities of these optimal algorithms, and for guiding the design of more robust bandit strategies (see also, \citet{Ashutosh2021RobustBandits}). 

In recent work, \citet{FanGlynn2024Fragility} establish a lower bound on the tail probabilities of the regret for any asymptotically optimal bandit algorithm, that is, one which matches the leading-order term in the regret lower bound, including the exact multiplicative constant. They further provide a detailed analysis of the regret tail upper bound for the KL-UCB algorithm of \citet{cappe2013kullback}, showing that its regret tail matches their lower bound, thereby demonstrating the tightness of these bounds for the single-parameter exponential family (SPEF) of reward distributions. 

While the authors discuss the extensions of their lower bounds to certain non-parametric classes of reward distributions, the regret tail upper bounds for optimized algorithms in these more general settings remained open. In particular, their analysis does not extend to the $\KLinf$-UCB algorithm, which is known to be optimal for finitely-supported distributions \citep[Theorem 2]{cappe2013kullback}, bounded-support distributions \citep[Proposition 4]{agrawal21heavytailed}, as well as for a broad class of heavy-tailed distributions \citep{agrawal21heavytailed}. This motivates the following central question:
\begin{quote} 
\emph{Can we characterize the regret-tail behavior of asymptotically optimal bandit algorithms for broader nonparametric classes of reward distributions?}
\end{quote}

In this work, we address the above question affirmatively. We first show that a natural generalization of the $\KLinf$-UCB algorithm is asymptotically optimal for a very broad class of distributions $\cal L$, introduced later. We then prove the first regret-tail upper bound for this algorithm for arm distributions from $\cal L$, and demonstrate its tightness in different settings. We conclude this section by presenting the key contributions of this work.

\begin{itemize}
    \item We extend the $\KLinf$-UCB algorithm of \citet{agrawal21heavytailed} to a much broader family of reward distributions $\mathcal L$ (to be introduced later), and show that it is asymptotically optimal (Theorem~\ref{theorem-1}). This is a substantial generalization that includes the moment-bounded class studied by the authors, as well as bounded-support distributions as special cases. 
    
    \item Our main technical result is the regret tail upper bound for the generalized $\KLinf$-UCB algorithm (Theorem~\ref{theorem-3}). This is the first regret-tail upper bound for asymptotically optimal bandit algorithms beyond parametric models. As immediate corollaries, we get the upper bounds for the classical settings: \textit{(a)} the moment-bounded heavy-tailed class considered in \citet{agrawal21heavytailed} (Corollary~\ref{cor-2}); and \textit{(b)} the bounded-support reward distributions (Corollary~\ref{cor-3}). 
    
    When the class $\mathcal{L}$ satisfies an additional property called discrimination-equivalence (Definition~\ref{def-5.3}), we can show that the upper bounds that we prove exactly match the regret tail lower bound proposed in \citet{FanGlynn2024Fragility}. 
    


    \item For the special yet nontrivial case of finitely-supported distributions, a class that does not satisfy the discrimination-equivalence property, we provide a much tighter regret tail upper bound (Theorem~\ref{theorem-4}), and show that it matches the lower bound of \citet{FanGlynn2024Fragility}.

\end{itemize}

\noindent\textbf{Organization:}
The rest of the paper is organized as follows. Section~\ref{background} presents the setup and necessary background. In Section~\ref{algorithm}, we present the \emph{Generalized} $\mathrm{KL}_{\inf}$-UCB algorithm and analyze its asymptotic optimality in the sense of expected regret. Section~\ref{fragility} analyzes the regret tail behavior of this optimal UCB algorithm, and presents a tight regret tail characterization under discrimination equivalence. Finally, Section~\ref{bounded-finite-support} provides a refined regret tail analysis for the finitely supported setting, to obtain a tight upper bound, without relying on the discrimination equivalence. We review the relevant literature in detail in Appendix~\ref{sec:lit}.

\section{Preliminaries}\label{background}
We begin by introducing the setup in the next section, followed by the general class $\cal L$ in Section~\ref{sec:calL}. 

\subsection{Multi-arm Bandit Setup}
We consider a bandit problem with $K$-arms indexed by $a \in [K]:=\{1, \ldots, K\}$, with $K \geq 2$. Let \(\mathcal{ P}(\Re)\) denote the collection of all probability measures on \(\Re\). For \(\nu \in\mathcal{P}(\Re)\), let \(m(\nu) = \int\nolimits_{\Re} x d\nu\) denote its mean. For \(\mathcal{M} \subseteq \mathcal{ P}(\Re)\), we denote a collection of bandit instances by \(\mathcal M^K\), a collection of vectors of \(K\) distributions, each from \(\mathcal{M}\). 

Next, given a bandit environment \( {\boldsymbol{\nu}} \in \mathcal M^K\) such that  \({\boldsymbol{\nu}} =  \{\nu_1,\nu_2,\dots,\nu_K\}\), each arm $a \in [K]$ has a reward distribution $\nu_a$ with expected reward $\mu_a = m(\nu_a)$. Let $\mu^* = \max \{\mu_a: a \in [K]\}$ be the expected reward of the optimal arm. We denote the sub-optimality gap of an arm $a$ as $\Delta_a = \mu^* - \mu_a$. Without loss of generality, we assume that arm 1 is the optimal arm, and means are ordered so that \(\mu_1 > \mu_2 > \ldots > \mu_K\) for the rest of the paper. 

At each time $t$, the algorithm selects an arm $A_t \in [K]$ based on the past information, and receives a reward $Y_t$. The number of times each arm~$a$ is pulled till time $t$ is denoted by $N_a(t)\triangleq \sum_{s=1}^{t} \mathbb{I}\{A_s = a\}$. In addition, for each arm $a$ and all rounds $t$ such that $N_a(t) \geq 1$, the empirical reward distribution of arm~$a$ is defined as $\hat{\nu}_a(t) = \frac{1}{N_a(t)} \sum_{s=1}^{t} \delta_{Y_s} \, \mathbb{I}\{A_s = a\}$. Here $\delta_{Y_s}$ denotes the Dirac measure at $Y_s$. The quality of an algorithm $\pi$ is evaluated using the standard notion of \emph{expected regret}, defined next. The expected regret till time $T \geq 1$ is defined as 

\begin{align*}
\mathbb{E}[R(T)] &\triangleq \mathbb{E} \left[ T\mu_1 - \sum_{t=1}^{T} Y_t \right] = \sum_{a=1}^{K} \Delta_a \, \mathbb{E}[N_a(T)].
\end{align*}

Note that the above expectation is with respect to the probability measure $\mathbb{P}^{\pi}_{{\boldsymbol{\nu}}}$ which is induced by the interaction between the algorithm $\pi$ and the bandit environment ${\boldsymbol{\nu}}$. We refer the reader to \citet[Chapter 4]{lattimore2020bandit} for the details of the underlying measure spaces. 

\citet{LaiRobbins85} showed that any algorithm must incur at least logarithmic regret with a precise leading constant, a result later generalized by \citet{BURNETAS1996}. Several algorithms are now known to achieve this bound. The definition of $\mathcal{M}^K$-optimality, introduced below, formalizes this idea by requiring that an algorithm exactly matches this lower bound.
\begin{definition}
An algorithm is said to be  $\mathcal{M}^K$-optimal if for every ${\boldsymbol{\nu}} \in \mathcal{M}^K$ and for each sub-optimal arm $a$, the following holds:
\begin{align}\label{eq:lai-robbins-lower-bound}
    \lim_{T \to \infty} \frac{\mathbb{E}[N_a(T)]}{\log T} = \frac{1}{\mathrm{KL}^{\mathcal{M}}_{\inf}(\nu_a, \mu_1)},
\end{align}
where $\mathrm{KL}^{\mathcal{M}}_{\inf}(\nu_a, \mu_1) := \inf \{\mathrm{KL}(\nu_a,\nu'): \nu' \in \mathcal{M} \text{ and } \mathbb{E}[\nu'] \geq \mu_1\}$ and $\mathrm{KL}(\nu_a,\nu') = \int\log{\frac{d\nu_a}{d\nu'}}d\nu_a$ is the Kullback-Leibler divergence between $\nu_a$ and $\nu'$. 
\end{definition}
More generally, for a probability distribution $\nu$ and $x\in\Re$, $\mathrm{KL}^{\mathcal{M}}_{\inf}(\nu,x) := \inf \{\mathrm{KL}(\nu,\nu'): \nu' \in \mathcal{M} \text{ and } \mathbb{E}[\nu'] \geq x\}$. The proof of the Lai-Robbins lower bound~\citep{LaiRobbins85, BURNETAS1996} relies on change of measure arguments. \citet[Lemma 1]{AgrawalJunejaGlynn2020OptimalHeavyTails} show that it is necessary to impose certain restrictions on class \(\mathcal{M} \), otherwise \(\mathrm{KL}^{\mathcal{M}}_{\inf}(\cdot,\cdot) = 0 \), which makes it impossible to achieve logarithmic regret. To this end, we focus on a specific subclass $\mathcal{L}$ satisfying certain regularity properties which are discussed in the following section. 
\begin{remark}
    We use $\mathcal{M}$ to denote the general class of distributions. When $\mathcal{M}$ satisfies Assumptions~\ref{prop-1} and~\ref{prop-2}, we denote the corresponding subclass by $\mathcal{L}$. For notational convenience, we drop the superscript in $\KLinf(\cdot,\cdot)$ when the underlying class is clear from context.
\end{remark}

\subsection{Distribution Class $\mathcal{L}$}\label{sec:calL}
Since logarithmic regret is impossible without restrictions on the class $\mathcal{M}$, we restrict our attention to subclasses $\mathcal{L}$ that satisfy the following property:
\begin{align*}
\nu  \in \mathcal{L} \implies \nu \text{  satisfies Assumptions~\ref{prop-1} and~\ref{prop-2} }.\numberthis\label{eq:calL}
\end{align*}
\begin{assumption}\label{prop-1}
    For $\nu \in \cal L$, let $\hat{\nu}_n$ be the empirical distribution of $n$ i.i.d.\ samples from $\nu$, having mean $m(\nu)$. Let $g(n)=C_1\log(1+n)+C_2$ for some $C_1,C_2>0$. Then, $\nu, \mathcal L$ satisfy:
    \begin{align}\label{eq-2} 
       \mathbb{P}\left(\exists n \in \mathbb{N}: n \mathrm{KL}^{\mathcal L}_{\inf}(\hat{\nu}_n,m(\nu)) - g(n) \geq x \right) \leq e^{-x}. 
    \end{align}
\end{assumption}

\begin{assumption}\label{prop-2}
    For $\nu \in \cal L$, let $\hat{\nu}_n$ be the empirical distribution of $n$ i.i.d.\ samples from $\nu$, having mean $m(\nu)$, and let $\delta>0$. Then there exist constants $d_0>0$ and $c_\nu > 0$ s.t. for all $d<d_0$ the following holds:
    \begin{align}\label{eq-3}
    \mathbb{P} \bigg( \mathrm{KL}^{\mathcal L}_{\inf}(\hat{\nu}_n, m(\nu)+ \delta) \leq \mathrm{KL}^{\mathcal L}_{\inf}(\nu, m(\nu)+ \delta) - d \bigg) 
    \leq e^{-n c_{\nu} d^2}.
    \end{align}
\end{assumption}
If $\mathrm{KL}^{\mathcal L}_{\inf}(\cdot, x)$ is appropriately continuous, then $\mathrm{KL}^{\mathcal L}_{\inf}(\hat{\nu}_n, m(\nu)) \to \mathrm{KL}^{\mathcal L}_{\inf}(\nu, m(\nu))$ and $\mathrm{KL}^{\mathcal L}_{\inf}(\hat{\nu}_n, x) \to \mathrm{KL}^{\mathcal L}_{\inf}(\nu, x) > 0$ for $x > m(\nu)$ as $n\to\infty$. Assumptions~\ref{prop-1} and~\ref{prop-2} correspond to an exponential concentration of the empirical $\mathrm{KL}^{\mathcal L}_{\inf}$ around these limits. 

To build intuition, recall that for $m, m' \in \mathbb{R}$, $\operatorname{KL}(N(m,1), N(m',1)) \propto (m-m')^2$, where $N(\cdot, 1)$ denotes the Gaussian distribution with unit variance and given mean. If $\mathcal L = \{N(m',1): m' > 0\}$, then $\KLinf(N(m,1), 0) \propto m^2$, since in this special case, there is a unique distribution in $\mathcal L$ with a given mean. Thus, inequality~\eqref{eq-2} can be viewed as a time-uniform exponential concentration bound for the deviation of the empirical mean from the true mean, analogous to Hoeffding's inequality. Several commonly encountered families of distributions, such as bounded support~\citep{HondaBounded10}, and the collection of distributions with a known uniform bound on their $(1+\epsilon)^{th}$ moment (for a known $\epsilon > 0$) ~\citep{agrawal21heavytailed}, are examples of $\cal L$. We discuss this further in Appendix~\ref{appendix:generality}. 

\section{Generalized $\KLinf$-UCB Algorithm and its optimality}\label{algorithm}
We now present a natural generalization of the $\KLinf$-UCB algorithm for arm distributions from $\cal L$, and prove that it is asymptotically optimal. Let ${\boldsymbol{\nu}} \in \mathcal{L}^K$ denote a $K$-armed bandit instance. 

\begin{algorithm}[!htbp]
\caption{\emph{Generalized} $\mathrm{KL}_{\inf}$-UCB\((K, \{f_a(\cdot)\}_{a=1}^K)\)} \label{Alg1}
\textbf{Input:}\(~K;~\) \( \cal L; \) exploration functions for each arm, i.e., \(f_a(\cdot)\). 

\textbf{Initialization:} Pull each arm $a \in [K]$ once
\newline
Set \(t \longleftarrow K+1 \).
\newline
Compute \(\hat{\nu}_a(t)\), and update \(N_a(t)\) for all arms \(a\in [K]\).
\begin{algorithmic}[1]\label{alg:KLinfUCB}
\FOR{$t = K+1, K+2, \dots$ }
    \FOR{each arm $a \in \{1, \ldots, K\}$}
        \STATE Compute index 
        $U_a(N_a(t),t) = \sup \left\{\Exp{\nu}: \nu \in \mathcal{L}, \mathrm{KL}(\hat{\nu}_a(t),\nu) \leq \frac{f_a(t)}{N_a(t)} \right\}$.
    \ENDFOR
    \STATE Pick an arm $A_{t+1} \in {\arg \max}_{a \in [K]}~U_a(N_a(t),t)$.
    \STATE Set \( t\leftarrow t + 1 \).
    \STATE Update \(\hat{\nu}_a(t)\), and update \(N_a(t)\) for all arms \(a\in [K]\).
\ENDFOR
\end{algorithmic}
\end{algorithm}



For exploration functions $f_a(\cdot)$ for each $a\in[K]$,  define the index of arm $a$ as
\begin{align*}
    U_a(N_a(t),t) = \sup \left\{x \in \Re:  \mathrm{KL}_{\inf}(\hat{\nu}_a(t),x) \leq \frac{f_a(t)}{N_a(t)} \right\}.
\end{align*}
Note that the index $U_a(\cdot,\cdot)$ is monotonically increasing in the second argument when the first argument is fixed. 
Computing the $\KLinf$-UCB index, while polynomial-time for many common distribution classes, can still be computationally demanding. To mitigate this, \citet{agrawal21heavytailed} proposed, in their setting, a batched variant with exponentially increasing batch sizes, at the cost of a slightly higher regret (within a constant factor of the lower bound). We instead focus on the unit-batch version of the algorithm and establish its asymptotic optimality over $\mathcal{L}$. Understanding the trade-off between statistical optimality and computational efficiency, and designing computationally efficient yet asymptotically optimal algorithms, remain interesting directions for future work.

\begin{theorem}\label{theorem-1}
	For \({\boldsymbol{\nu}} \in \mathcal L^K \) and for every suboptimal arm $a$, let \( f_a(t) = \log(t) + 2\log\log(t) + g(N_a(t)) \). Then Generalized \(\mathrm{KL}_{\inf} \)-UCB, with inputs $(K, f_a(.))$ satisfies
	\begin{equation*}
	\limsup\limits_{T\rightarrow \infty}\frac{\mathbb{E}[N_a(T)]}{\log(T)} \leq \frac{1}{\mathrm{KL}_{\inf}(\nu_a,\mu_1)}. 
	\end{equation*}
\end{theorem}

The proof of the above theorem closely follows that of \citet[Theorem 1]{agrawal21heavytailed} and is presented in the Appendix~\ref{appendix:proof_theo1} for completeness. 

\section{Regret Tail Characterization}\label{fragility}
We now discuss the regret-tail characterization for the optimal (Generalized) $\KLinf$-UCB.
\subsection{Regret Tail Lower Bound}
Since optimal bandit algorithms explore just enough to achieve good \emph{expected regret}, they are fragile in that large-regret events can still occur with a non-negligible probability. \citet{FanGlynn2024Fragility} formalized this and proved a lower bound on the probability that regret is large for optimal bandit algorithms. We formally present this lower bound below. Note that the lower bound holds for any generic distribution class $\mathcal{M}$ of arms; not just for classes $\cal L$ of arm distributions.

\begin{theorem}\label{theorem-2}
    Let $\pi$ be $\mathcal{M}^K$-optimal algorithm, and for $\gamma \in (0,1)$, let $\mathcal{D}_\gamma(T)=\bigl[\log^{1+\gamma}(T),\, (1-\gamma)T \bigr]$. Then, for ${\boldsymbol{\nu}} \in \mathcal{M}^K$, and for an $i$-th best arm,
    \begin{equation}\label{eq:tail-lower-bound}
    \liminf_{T \to \infty}
    \inf_{x \in \mathcal{D}_\gamma(T)}
    \frac{\log \mathbb{P}_\nu^\pi\!\left(N_{i}(T) > x\right)}{\log x}
    \;\ge\; - \sum_{j=1}^{i-1}
    \inf_{\substack{\widetilde{\nu} \in \mathcal{M}:\\ m(\widetilde{\nu}) < \mu_i}}
    \frac{\mathrm{KL}\!\left(\widetilde{\nu}, \nu_j\right)}
    {\mathrm{KL}_{\inf}\!\left(\widetilde{\nu}, \mu_i\right)}. 
    \end{equation}
\end{theorem}
\citet[Section~3.2]{FanGlynn2024Fragility} prove the above result for the special case when $\mathcal{M}^K$ is SPEF, and state the more general result as in the Theorem above. Closely following their proof for SPEF, we explicitly provide a proof for the general lower bound in Appendix~\ref{appendix:proof_theo2} for completeness.  

Note that each term inside the summation on the right-hand side of the Equation~\eqref{eq:tail-lower-bound}, denoted as  $C(\widetilde{\nu},{\nu_j}) := \tfrac{\mathrm{KL}(\widetilde{\nu}, \nu_j)}{\mathrm{KL}_{\inf}(\widetilde{\nu}, \mu_i)}$ is at least 1, i.e., $C(\widetilde{\nu},{\nu_j}) \geq 1$. This directly follows from the definition of $\KLinf(\cdot,\cdot)$ with $\mu_j > \mu_i$. For the special case when $\inf_{\widetilde{\nu}} C(\widetilde{\nu},{\nu_j})=1$ for every sub-optimal arm $j$, the condition that holds when the underlying model class $\mathcal{M}$ is discrimination equivalent. Now, we formally define discrimination equivalence property in the following section.


\subsection{Discrimination Equivalence}
Discrimination equivalence is a key structural property of a distribution class that governs regret tail behavior. As shown by \citet{FanGlynn2024Fragility}, under discrimination equivalence, even asymptotically optimal algorithms for SPEF models exhibit heavy-tailed regret, with tails resembling those of a truncated Cauchy distribution. In particular, the regret tail probability $\mathbb{P}(R(T) > x)$ decays only polynomially, at a rate of order $\approx x^{-1}$ over a wide range of deviations $x$. This implies that discrimination-equivalent classes are inherently hard to learn: even optimal algorithms incur large regret with non-negligible probability. The definition below generalizes that of discrimination equivalence for SPEF classes, introduced by \citet[Definition 3]{FanGlynn2024Fragility}. 
\begin{definition}\label{def-5.3}
    We say a distribution class $\mathcal{M}$ is discrimination equivalent if for every pair
    $\nu,\nu' \in \mathcal{M}$ with $m(\nu) > m(\nu')$, the following holds
    \begin{equation}\label{eq:de}
    \inf_{\substack{\widetilde{\nu} \in \mathcal{M}:\\ m(\widetilde{\nu}) < m(\nu')}}
    \frac{\mathrm{KL}\!\left(\widetilde{\nu}, \nu\right)}
    {\mathrm{KL}_{\inf}\!\left(\widetilde{\nu}, m(\nu')\right)} = 1. 
    \end{equation}
\end{definition}
As noted earlier, the left-hand side of~\eqref{eq:de} is always at least $1$. Indeed, $\nu$ itself is a feasible candidate in the optimization defining $\KLinf(\tilde{\nu}, m(\nu'))$, since $m(\nu) > m(\nu')$.

To build intuition for condition in~\eqref{eq:de}, suppose the optimization defining $\KLinf(\tilde{\nu}, m(\nu'))$ admits a solution $\nu^*$. Then the ratio $\mathrm{KL}(\widetilde{\nu},\nu)/\mathrm{KL}(\widetilde{\nu},\nu^*)$ can be interpreted as the relative difficulty of two sequential testing problems:  ($\widetilde{\nu} \textbf{ vs. } \nu$) and ($\widetilde{\nu} \textbf{ vs. } \nu^*$), given samples from $\widetilde{\nu}$. If there is a distribution $\tilde{\nu}\in \mathcal M$ for which these two problems are comparably difficult, then the ratio is close to one, indicating the class $\cal M$ is intrinsically hard to learn.

Exploiting the special formulation of $\mathrm{KL}$-divergence for SPEF,~\citet[Lemma~1, Proposition~1 in Section~3.1]{FanGlynn2024Fragility} provide several characterizations and examples of discrimination equivalent subclasses of SPEF. However, extending such characterizations to general distributions is not straightforward. In fact, we are unaware of any examples beyond those provided by the authors. In contrast, here we show that the two commonly studied non-parametric classes (also the running examples for this work) --- \emph{moment-bounded class and the bounded-support distributions --- are not discrimination equivalent.}  

\begin{example}\label{exmpl:moment-bounded}
    In this example, we consider the class $\mathcal L_{B,\varepsilon}$ of distributions whose $(1+\varepsilon)$-th moments are uniformly bounded by a constant $B$. 
    
    For $\varepsilon>0$ and $B>0$, define 
    \[\mathcal{L}_{B,\varepsilon} \;\triangleq\; \left\{ \nu \in \mathcal{P}(\mathbb{R}) :\; \mathbb{E}_{\nu}\!\left[|X|^{1+\varepsilon}\right] \le B \right\}. \numberthis\label{eq:Lbeps} \] 
    Below, we demonstrate a pair of distributions $\nu, \nu' \in \mathcal L_{B,\varepsilon}$ such that the condition in Equation~\eqref{eq:de} is violated. For simplicity, let's take $\epsilon=1$ and $B=1$.
    
    Let $\nu = \delta_{1}$ be the point mass at $\{1\}$ so that $m(\nu)=1$, and let $\nu'=\mathrm{Bernoulli}(p)$ for some $p\in(0,1)$. Clearly, $\nu,\nu' \in \mathcal{L}_{1,1}$ with  $m(\nu)>m(\nu')$. Note that $\widetilde{\nu} = \delta_{1}$ is not permissible because $m(\widetilde{\nu}) > m(\nu')$. Hence any arbitrary distribution $\widetilde{\nu} \in \mathcal{L}_{1,1}$ if supported at any point other than 1, then $\mathrm{KL}(\widetilde{\nu},\nu') = +\infty$. However, from \citet[Lemma~4.2]{Agrawal2022HeavyTailsBandits}, we know that $\mathcal{L}_{1,1}$ is a compact set in the topology of weak convergence induced by the \emph{L\'evy metric}. Also, $\KLinf(.,m(\nu'))$ is a continuous function in this weak topology~\citep[Lemma~4.11]{Agrawal2022HeavyTailsBandits}. We refer the reader to section ~4.2 and section~4.3 in~\citet{Agrawal2022HeavyTailsBandits} for more details. Thus, using these arguments, we can conclude that $\KLinf(.,m(\nu'))$ is bounded. Hence,
    \[ 
    \underset{\widetilde{\nu} :m(\widetilde{\nu}) < m(\nu')}{\inf}~~
    \frac{\mathrm{KL}(\widetilde{\nu},\nu)}
    {\mathrm{KL}_{\inf}(\widetilde{\nu},m(\nu'))}
    \;=\;+\infty.
    \]
\end{example}

\begin{example}\label{exmpl:bounded-support}
    We now consider the class of reward distributions with bounded support. For a constant $a,b \in \Re$, define 
    \[\mathcal{L}_{[a,b]} \;\triangleq\; \left\{ \nu \in \mathcal{P}(\mathbb{R}):\; \operatorname{Supp}(\nu) \subseteq [a,b] \right\}. \numberthis\label{eq:Lbdd} \] 
    For simplicity of presentation, without loss of generality, we take $a=0, b=1$ so that 
    $\mathcal{L}_{[0,1]} = \{\nu \in \mathcal{P}([0,1])\}.$
    Note that the example from Example~\ref{exmpl:moment-bounded} serves as a counter example for this class as well. However, for providing more intuition, we construct another example, demonstrating that $\mathcal L_{[0,1]}$ is not discrimination equivalent. 
    
    Let ${\nu}=\mathrm{Bernoulli}(p)$ and $\nu' = \mathrm{Bernoulli}(q)$ with $p,q \in (0,1)$ such that $q<p$. Note that if $\widetilde{\nu}$ is supported at any point other than $\{0\}$ and $\{1\}$, then $\mathrm{KL}(\widetilde{\nu},\nu) = +\infty$. Thus, we restrict our attention, where $\widetilde{\nu}$ is supported on both $\{0,1\}$ or any one of $\{0\}$ or $\{1\}$. If supported on both $\{0,1\}$ then $\widetilde{\nu} = \mathrm{Bernoulli}(r)$ with $r \in (0,1)$ such that $r<q$ due to the mean restriction.   
    Consequently, the ratio
    \[
    \frac{\mathrm{KL}(\widetilde{\nu},\nu)}
    {\mathrm{KL}_{\inf}(\widetilde{\nu},m(\nu'))}
    \;=\;
    \frac{\mathrm{KL}(\mathrm{Bernoulli(r)},\mathrm{Bernoulli}(p))}
    {\mathrm{KL}(\mathrm{Bernoulli(r)},\mathrm{Bernoulli}(q))} 
    \neq 1.
    \]
    If $\widetilde{\nu}$ is supported in any one of $\{0\}$ or $\{1\}$, then $\widetilde{\nu}$ has to be $\delta_0$. Otherwise, if $\widetilde{\nu}=\delta_1$ then $m(\widetilde{\nu}) > m(\nu')$. Thus, taking $\widetilde{\nu} = \delta_0$, similar to the above, the ratio is not equal to 1. Hence, the bounded-support distribution class $\mathcal{L}_{[0,1]}$ is \emph{not} discrimination
equivalent.
\end{example}



\subsection{Regret Tail Upper Bound}

\citet{FanGlynn2024Fragility} establish a regret tail upper bound for the KL-UCB algorithm, which is asymptotically optimal for SPEF models. However, their analysis does not extend beyond the SPEF setting to more general distribution classes. In particular, it does not yield guarantees for the $\KLinf$-UCB algorithm, which is asymptotically optimal for broader families, including bounded, finitely supported, and certain heavy-tailed distributions. Motivated by this gap, in the following, we present our main result: a regret tail upper bound for the generalized $\KLinf$-UCB algorithm (Algorithm~\ref{alg:KLinfUCB}), applicable to distribution classes beyond the SPEF setting.\begin{theorem}\label{theorem-3}
    Let $\pi$ be $\mathcal{L}^K$-optimal Generalized $\mathrm{KL}_{\inf}$-UCB algorithm. For $\gamma \in (0,1)$, let $\mathcal{D}_\gamma(T)=\bigl[\log^{1+\gamma}(T),\, (1-\gamma)T \bigr]$. Then, for any ${\boldsymbol{\nu}} \in \mathcal{L}^K$, and for the $i$-th best arm in ${\boldsymbol{\nu}}$,
    \begin{equation}\label{eq-4}
    \limsup_{T \to \infty}
    \inf_{x \in \mathcal{D}_\gamma(T)}
    \frac{\log \mathbb{P}_{{\boldsymbol{\nu}}}^\pi\!\left(N_{i}(T) > x\right)}{\log x}
    \;\le\; - (i-1).
    \end{equation}
\end{theorem}

Theorem~\ref{theorem-3} provides the first regret tail upper bound for asymptotically optimal bandit algorithms beyond parametric models. This result has two important implications, stemming from the generality of the distribution class $\mathcal{L}$ and the flexibility of the $\mathrm{KL}_{\inf}$-UCB framework, which we discuss next. We refer the reader to  Appendix~\ref{appendix:proof_theo3} for a proof. 

\noindent\textbf{Upper bound for moment-bounded models.} Recall that each element of the class $\mathcal L_{B,\varepsilon}$ defined in Equation~\eqref{eq:Lbeps} satisfies Assumptions~\ref{prop-1} and~\ref{prop-2} (Appendix~\ref{appendix:generality}). Hence, Theorem~\ref{theorem-1} implies that the $\mathrm{KL}_{\inf}$-UCB algorithm is asymptotically optimal for bandits with arm distributions from $\mathcal{L}_{B,\varepsilon}$, originally proven by \citep{agrawal21heavytailed}. Thus, we get the following corollary to Theorem~\ref{theorem-3}, 


\begin{corollary}\label{cor-2}
    Let $\pi$ be $\mathcal{L}^K_{B,\varepsilon}$-optimal \textit{Generalized} $\mathrm{KL}_{\inf}$-UCB algorithm. For $\gamma \in (0,1)$ let $\mathcal{D}_\gamma(T)=[\log^{1+\gamma}(T),\, (1-\gamma)T ]$. Then for each ${\boldsymbol{\nu}} \in \mathcal{L}^K_{B,\varepsilon}$, $\pi$ satisfies the regret tail upper bound in Equation~\eqref{eq-4}. 
\end{corollary}


\noindent\textbf{Upper-bound for bounded-support models.} As in the previous example, recall that each element $\nu\in\mathcal L_{[a,b]}$, defined in Equation~\eqref{eq:Lbdd}, also satisfies Assumptions~\ref{prop-1} and~\ref{prop-2}. Again, the \textit{Generalized} $\mathrm{KL}_{\inf}$-UCB algorithm  with $f_a(t)=\log t + \log\log t$ is asymptotically optimal for the bounded-support model $\mathcal{L}_{[a,b]}$ (also see, \citet{cappe2013kullback,agrawal21heavytailed}). Thus, we get the following corollary to Theorem~\ref{theorem-3}, 


\begin{corollary}\label{cor-3}
    Let $\pi$ be $\mathcal{L}^K_{[a,b]}$-optimal Generalized $\KLinf$-UCB algorithm. For $\gamma \in (0,1)$ let $\mathcal{D}_\gamma(T)=[\log^{1+\gamma}(T), (1-\gamma)T ]$. Then for each ${\boldsymbol{\nu}} \in \mathcal{L}^K_{[a,b]}$, $\pi$ satisfies the regret tail upper bound in Equation~\eqref{eq-4}.
\end{corollary}
Corollaries~\ref{cor-2} and~\ref{cor-3} of Theorem~\ref{theorem-3} follow since Generalized $\KLinf$-UCB with appropriate thresholds $f_a$ (with functions $g$ chosen as in Appendix~\ref{appendix:generality}), is optimal for the  two classes. 

\noindent\textbf{Tight Upper Bound with Discrimination Equivalence.} 
From Theorems~\ref{theorem-2} and~\ref{theorem-3}, we observe that for the class $\mathcal{L}$ of distributions satisfying~\eqref{eq:calL}, the regret tail lower and upper bounds do not always match. However, when $\mathcal{L}$ is additionally discrimination equivalent, the lower bound in~\eqref{eq:tail-lower-bound} simplifies to $-(i-1)$, exactly matching the upper bound in~\eqref{eq-4}. Thus, for discrimination-equivalent classes, the regret tail upper bound in Theorem~\ref{theorem-3} is tight.

This aligns with the intuition that discrimination-equivalent classes are hard to learn. The following corollary formalizes this and shows that even the optimal $\KLinf$-UCB algorithm exhibits heavy regret tails for such classes.


\begin{corollary}\label{cor-1}
    Let $\mathcal L_{DE}$ be a discrimination equivalent class that also satisfies the property in Equation~\eqref{eq:calL}. Let $\pi$ be the Generalized $\mathrm{KL}_{\inf}$-UCB algorithm for $\mathcal L_{DE}$. For $\gamma \in (0,1)$, let $\mathcal{D}_\gamma(T)=[\log^{1+\gamma}(T), (1-\gamma)T]$. Then, for each ${\boldsymbol{\nu}} \in \mathcal{L}^K_{DE}$, $\pi$ satisfies the following:
\end{corollary}
\begin{equation}
    \lim_{T \to \infty}
    \inf_{x \in \mathcal{D}_\gamma(T)}
    \frac{\log \mathbb{P}_{{\boldsymbol{\nu}}}^\pi\!\left(N_{i}(T) > x\right)}{\log x}
    \;=\; - (i-1).
\end{equation}

While we provide the first regret-tail upper bound for asymptotically optimal UCB algorithms over a broad class $\cal L$ beyond SPEFs, the upper bound in Theorem~\ref{theorem-3} is tight only for discrimination-equivalent subclasses $\mathcal L_{DE}$. In that case, the number of pulls of the second-best arm satisfies $\mathbb{P}^{\pi}_{\boldsymbol{\nu}}(N_2(T)>x) \approx x^{-1}$, which in turn implies $\mathbb{P}^{\pi}_{\boldsymbol{\nu}}(R(T)>x) \approx x^{-1}$, yielding Cauchy-like regret tails. By contrast, the moment-bounded class $\mathcal{L}_{B,\varepsilon}$ and the bounded-support class $\mathcal{L}_{[a,b]}$ are not discrimination equivalent (Examples~\ref{exmpl:moment-bounded} and~\ref{exmpl:bounded-support}), and so our current upper bound is not tight for these classes. We note that discrimination equivalence is sufficient but not necessary for tightness. The finitely-supported class, discussed in the following section, demonstrates that tight regret-tail behavior can arise even when this condition fails. Obtaining sharper regret-tail upper bounds for general classes $\cal L$ for the $\cal L$-optimized $\KLinf$-UCB remains an interesting direction for future work. 


\section{Tight Regret Tail Upper Bound for Finitely Supported Models}\label{bounded-finite-support}
In the preceding section, we provide the first regret tail characterization for a broad class of distributions $\mathcal{L}$. While our analysis is very general and can handle several non-parametric classes, our bounds are tight only when the class is also {discrimination equivalent}.  
In this section, we provide a much tighter analysis and upper bound for the special yet nontrivial case when $\cal L$ is the collection of distributions supported in a known support set. This corresponds to a finite-alphabet model, which allows the use of finite-sample large deviation bounds like Sanov-type inequalities.

Formally, let $\mathcal{L}_{[a,b],\mathcal{X}}$ denote the class of distributions supported on at most known and fixed $s<\infty$ points. Let $\mathcal X := \{x_1, \dots, x_s\} \subset [a,b]$ be known and fixed. Then, 
\[\mathcal{L}_{[a,b],\mathcal{X}} \;\triangleq\; \left\{\nu \in \mathcal{P}(\Re) : \operatorname{Supp}(\nu) \subseteq \mathcal X \right\}.\] 

\begin{theorem}[\textbf{Tight regret tail upper bound for finitely-supported class}]\label{theorem-4}
    Let $\pi$ be $\mathcal{L}^K_{[a,b],\mathcal{X}}$-optimal $\KLinf$-UCB algorithm with $f_a(t)=\log t + \log\log t$. For $\gamma \in (0,1)$, let $\mathcal{D}_\gamma(T)=[\log^{1+\gamma}(T),\, (1-\gamma)T ]$. Then, for each ${\boldsymbol{\nu}} \in \mathcal{L}^K_{[a,b],\mathcal{X}}$, and for the $i$-th best arm in $\boldsymbol{\nu}$,
    \begin{equation*}
        \lim_{T \to \infty}
    \inf_{x \in \mathcal{D}_\gamma(T)}
    \frac{\log \mathbb{P}_{{\boldsymbol{\nu}}}^\pi\!\left(N_{i}(T) > x\right)}{\log x}
    = - \sum_{j=1}^{i-1}
    \inf_{\substack{\widetilde{\nu} \in \mathcal{L}_{[a,b],\mathcal{X}}:\\ m(\widetilde{\nu}) < \mu_i}}
    \frac{\mathrm{KL}\!\left(\widetilde{\nu}, \nu_j\right)}
    {\mathrm{KL}_{\inf}\!\left(\widetilde{\nu}, \mu_i\right)}.
    \end{equation*}
\end{theorem}

\noindent \textit{Proof Outline:} The proof proceeds by decomposing the event $\mathbb{P}_{{\boldsymbol{\nu}}}^{\pi}\!\left(N_i(T) > x\right)$ into two complementary events. The first is controlled by partitioning according to the number of pulls of sub-optimal arms and a union bound, each of whose probabilities are bounded using a finite-sample version of Sanov’s theorem~\citep{dembo2010large, Csiszar2006Sanov}. This yields sharp exponential bounds. The second component is controlled using Assumption~\ref{prop-2} (finite-support models are examples of the model class $\mathcal{L}$). Combining these yields the desired bound. See Appendix~\ref{appendix:proof_theo4} for a complete proof.

\section{Conclusion, Limitations and Future Works}
In this work, we investigated the regret-tail behavior of stochastic multi-armed bandit algorithms that are optimal over a broad class of reward distributions, beyond SPEF classes.  Our results significantly extend earlier regret tail upper bounds previously established for SPEF models to much broader nonparametric classes, including moment-bounded and bounded-support distributions. We extend the $\mathrm{KL}_{\inf}$-UCB algorithm to a substantially more general class of reward distributions and establish its asymptotic optimality over this class. We further demonstrate the tightness of the proposed regret-tail upper bounds, both in the presence and in the absence of discrimination equivalence. The generality of our framework enables us to characterize regret-tail phenomena of a wide range of optimal algorithms across broad families of nonparametric bandit models. Despite these advances, several limitations remain. For the general distribution class $\mathcal{L}$, a gap persists between the regret-tail upper and lower bounds when discrimination equivalence does not hold. Closing this gap and designing algorithms that achieve improved tail robustness constitute important directions for future work. Further, the study of fragility issues in the Thompson sampling~\citep{SAgrawal2011TS1,riou2020bandit} remains a promising direction.

\section*{Acknowledgments}
Subhodip, a Ph.D student at the ECE department, is supported by Ministry of Education (MoE) fellowship. Shubhada acknowledges the generous support from the Pratiksha Trust, Bangalore, through the Young Investigator Award, and the DST INSPIRE Faculty Grant IFA24-ENG-389.


\bibliography{references}
\bibliographystyle{apalike}

\appendix

\section{Appendix}
\subsection{Generality of Distribution Class $\mathcal{L}$}\label{appendix:generality}
Recall from Section~\ref{sec:calL}, that class $\mathcal{L}$ is a collection of distributions that satisfies the following:
\[
\nu  \in \mathcal{L} \implies \nu \text{  satisfies Assumptions~\ref{prop-1} and~\ref{prop-2}}.\]

Two commonly studied distribution classes that are examples of $\mathcal{L}$ are presented below. 

\noindent\ \textbf{Moment bounded distributions:} For $\varepsilon>0$ and $B>0$, define class $\mathcal L_{B,\varepsilon}$ as:
\[
\mathcal{L}_{B,\varepsilon}
\;\triangleq\; \left\{ \nu \in \mathcal{P}(\Re)
:\; \mathbb{E}_{\nu}\!\left[|X|^{1+\varepsilon}\right] \le B \right\}.
\]
Every element $\nu \in \mathcal L_{B,\varepsilon}$ satisfies Assumptions~\ref{prop-1} and~\ref{prop-2} with $g(n) = 2\log(1+n) + 1$ (see, \citet[Proposition~5]{agrawal21heavytailed} and \citet[Lemma~6]{agrawal21heavytailed} for details).

\noindent\ \textbf{Bounded support distributions:}
For constants $a,b \in \Re$, define class $\mathcal L_{[a,b]}$ as: 
\[
\mathcal{L}_{[a,b]} \;\triangleq\;
\left\{ \nu \in \mathcal{P}(\Re):\;
\operatorname{Supp}(\nu) \subseteq [a,b] \right\}.
\]
This distribution class also satisfies the above assumptions with $g(n) = \log(1+n) + 1$. The proof follows along the same lines as in \citet{agrawal21heavytailed}, with an appropriate modification that employs the dual representation of $\KLinf$ for bounded-support distributions as developed in \citet[Section~4]{HondaBounded10} along with \citet[Lemma F.1]{agrawal2021optimal}.

\subsection{Proof of Theorem~\ref{theorem-1}}\label{appendix:proof_theo1}

This proof follows along the lines of \citet[Theorem~1]{agrawal21heavytailed}, specialized to batches of size one per trial. Take \( t \ge K+1 \), and without loss of generality, assume that arm~1 is optimal. 
Then the event that, at the beginning of round \(t\), a sub-optimal arm \(a \neq 1\) attains the maximum index i.e., the event \(\{A_t = a\}\) is contained in
\begin{align} \label{eq-5}
\{U_1(N_1(t),t) \leq \mu_1 ~~\text{ and }~~ A_{t}=a \}~ \bigcup ~\{U_a(N_a(t),t) > \mu_1 ~~\text{ and }~~ A_{t}=a \}.
\end{align}
The left-hand event of Equation~\eqref{eq-5} characterizes the underestimation of the optimal arm’s index relative to its true mean at time $t$, while the right-hand event corresponds to an overestimation of the sub-optimal arm’s index beyond the mean of the optimal arm. Recall that during the initial $K$ rounds, each arm is played exactly once as part of the initialization procedure. Thus,
\[ N_a(T) = 1+ \sum_{t=K+1}^T \mathbb{I}\{A_{t}=a\}, \]
$$\mathbb{E}[N_a(T)] \leq  ~ 1 + \mathbb{E}[D_T] + \mathbb{E}[E_T].$$ 
The terms $D_T$ and $E_T$ are as follows.
\[ D_T := \sum\limits_{t=K+1}^{T} \mathbb{I}\lrp{{U_1(N_1(t),t) \leq \mu_1, ~ A_{t} = a}}, ~\text{and }~ E_T := \sum\limits_{t=K+1}^{T} \mathbb{I}({U_a(N_a(t),t)> \mu_1, ~ A_{t}=a}). \]

\noindent \textbf{Bounding the overestimation of sub-optimal arms in $E_T$:}  
By the definition of the index employed by the algorithm, for any \(t\ge K+1\) and \(x\in\mathbb{R}\), the event  
\(\{U_a(N_a(t),t)\ge x\}\) is equivalent to  
\(\{N_a(t)\,\mathrm{KL}_{\inf}(\hat{\nu}_a(t),x)\le f_a(t)\}\).  
Let 
$$ 0 < d \le \min_{a>1}\mathrm{KL}_{\inf}(\nu_a,\mu_1).$$ Then the indicator of the event  
\(\{U_a(N_a(t),t)\ge \mu_1,\; A_t=a\}\) is dominated by the sum of the two events \(E_{1t}\) and \(E_{2t}\) defined below:

\[E_{1t} = \mathbb{I}\lrp{\mathrm{KL}_{\inf}({\hat{\nu}_a(t), \mu_1}) \leq \frac{f_a(t)}{N_a(t)},~\mathrm{KL}_{\inf}({\hat{\nu}_a(t), \mu_1}) > \mathrm{KL}_{\inf}(\nu_a,\mu_1)-d,~ A_{t} = a }, \]
\[E_{2t} = \mathbb{I}\lrp{\mathrm{KL}_{\inf}({\hat{\nu}_a(t), {\mu_1}}) \leq \mathrm{KL}_{\inf}(\nu_a,{\mu_1})-d, ~ A_{t}=a }.\]
Thus, 
\begin{align*}
E_T \le &\sum\limits_{t=K+1}^{T} E_{1t}+\sum\limits_{t=K+1}^{T}E_{2t}.
\end{align*}
We can clearly see that \(E_{1t}\), is bounded above by $\mathbb{I}\lrp{N_a(t)\lrp{\mathrm{KL}_{\inf}({\nu_a,\mu_1})-d} \leq f_a(t), A_{t}=a  },$  giving 
\begin{align}\label{eq-19}
    \sum\limits_{t=K+1}^{T}E_{1t} \leq \sum\limits_{t=1}^{T}E_{1t} \leq \sum\limits_{t=1}^{T}\mathbb{I}\lrp{{N_a(t)} \leq \frac{f_a(t)}{\mathrm{KL}_{\inf}({\nu_a,\mu_1})-d}, A_{t}=a  }.
\end{align}

Now, in order to upper bound the RHS of Equation~\eqref{eq-19}, we use the following lemma from \citet{agrawal21heavytailed}. For completeness, we state the result as follows.
\begin{lemma}\label{lem:BoundingSumE1j}For \( T\ge K+1 \), and \( d>0 \), 
	\[
	\sum\limits_{j=1}^{T} E_{1j} \le \lrp{\frac{\log(T)}{\mathrm{KL}_{\inf}(\nu_a,\mu_1)-d}+ O\lrp{\log\log(T)}}. 
	\]
\end{lemma}

We use the above Lemma \ref{lem:BoundingSumE1j} to get the following bound:  
\[		\sum\limits_{t=1}^{T}E_{1t}\le {\frac{\log(T)}{\mathrm{KL}_{\inf}(\nu_a,\mu_1)-d}+ O({\log\log(T)})}.
\]
For the exact form of the \( O(\log\log(T)) \) term above, we refer the reader to the proof of \citet[Lemma 14]{agrawal21heavytailed}.




Note that for any constant $c>0$, $1-e^{-c} \geq \frac{c}{1+c}$ because $e^c \geq 1+c$. Now in order to control the events in $E_{2t}$ we directly use Assumption~\ref{prop-2} with $\delta=\mu_1-\mu_a$ to get an upper bound as follows.
\begin{align*}
    \Exp{\sum\limits_{t=K+1}^{T}E_{2t}} &\leq \sum\limits_{t=K+1}^{T} \mathbb{P}(\mathrm{KL}_{\inf}({\hat{\nu}_a(t), {\mu_1}})  \leq \mathrm{KL}_{\inf}(\nu_a,{\mu_1})-d) \\ 
    &\leq \sum\limits_{t=K+1}^{T} e^{-tc_{\nu_a}d^2} \\
    &\leq \sum\limits_{t=2}^{\infty} e^{-tc_{\nu_a}d^2} \\
    &= \frac{e^{-2c_{\nu_a}d^2}}{1-e^{-c_{\nu_a}d^2}} \\
    &\leq 1+ \frac{1}{c_{\nu_a} d^2}.
\end{align*}

\noindent\textbf{Bounding underestimation of the optimal arm in \(D_T\):} This term contributes only a constant amount to the regret up to time $T$. To bound it, we can use \citet[Lemma 18]{agrawal21heavytailed}, which we restate below for completeness. Here we use Assumption~\ref{prop-1} for the proof of the following lemma.
\begin{lemma}\label{lem:BoundingDN}
	For \( T > K \), 
	\[ \Exp{D_T}  \le \frac{1}{\log(K+1)}.
	\]
\end{lemma}
\begin{proof}[Proof of Lemma~\ref{lem:BoundingDN}]
    Let's denote the time when arm $a$ is played for $r^{th}$ time by $T^r_a$. Note that $T^r_a \geq K+r-1$ because arm $a$ is played for $r-1$ rounds after each arm has been played once. Now $D_T$ can be rewritten as follows:
    \begin{align*}
         D_T &= \sum^{N_a(T)}_{r=2} \mathbb{I}\lrp{U_1(N_1(T^r_a),T^r_a) \leq \mu_1} \\
    &\leq  \sum^{T}_{r=2} \mathbb{I}\lrp{U_1(N_1(T^r_a),T^r_a) \leq \mu_1}
    \end{align*}
    \begin{align*}
        \mathbb{E}[D_T] &\leq \sum^{T}_{r=2} \mathbb{P}\lrp{U_1(N_1(T^r_a),T^r_a) \leq \mu_1} \\
        &= \sum^{T}_{r=2} \mathbb{P}\lrp{\KLinf(\hat\nu_1(T^r_k),\mu_1) \geq \frac{f_1(T^r_a)}{N_1(T^r_a)}} \\
        &= \sum^{T}_{r=2} \mathbb{P}\lrp{N_1(T^r_a)\KLinf(\hat\nu_1(T^r_a),\mu_1) \geq \log(T^r_a) + 2\log\log(T^r_a) + g(N_1(T^r_a)) } \\
        &\leq \sum^{T}_{r=2} \mathbb{P}\bigg(N_1(T^r_a)\KLinf(\hat\nu_1(T^r_a),\mu_1) - g(N_1(T^r_a)) \\
        &\quad \hspace{2cm}\geq \log(K+r-1) + 2\log\log(K+r-1) \bigg) \\
        &\leq \sum^{T}_{r=2} \frac{1}{K+r-1} \frac{1}{(\log(K+r-1))^2} \numberthis \label{eq:res} \\
        & \leq \int^\infty_2 \frac{1}{K+r-1} \frac{1}{(\log(K+r-1))^2} dr \\
        &= \frac{1}{\log(K+1)}
    \end{align*}
Equation~\eqref{eq:res} is a direct implication of Assumption~\ref{prop-1}.
\end{proof}
 Combining everything, we get 
\begin{align}\label{eq-7}
    \Exp{N_a(T)} \le  {\frac{\log(T)}{\mathrm{KL}_{\inf}\lrp{\nu_a,\mu_1}-d} +O(\log\log(T))+ \frac{1}{c_{\nu_a} d^2} + \frac{1}{(\log(K+1))} +2 .}
\end{align} 
The above bound in Equation~\eqref{eq-7} can be optimized over \( d \). Setting \( d \) to as below: \[ d = \lrp{c'_{\nu_a} \frac{(\mathrm{KL}_{\inf}(\nu_a,\mu_1))^2 }{ \log T}}^{1/3}, \] where \(c'_{\nu_a} = 2o(1)/c_{\nu_a} \) we get that 
\[\mathbb{E}[{N_a(T)}] \leq {\frac{\log T}{\mathrm{KL}_{\inf}(\nu_a,\mu_1)} + \frac{3(\log T)^{2/3} (c'_{\nu_a})^{1/3}}{2(\mathrm{KL}_{\inf}(\mu_a,m))^{4/3}} + O((\log  T)^{1/3}) + O(\log\log(T))}. \]
Finally, taking the limit, we get,

\[
\limsup\limits_{T\rightarrow \infty}\frac{\mathbb{E}[N_a(T)]}{\log(T)} \leq \frac{1}{\mathrm{KL}_{\inf}(\nu_a,\mu_1)}.
\] 

\subsection{Proof of Theorem~\ref{theorem-2}}\label{appendix:proof_theo2}
\noindent \textbf{Two-arm setting:} We first establish the lower bound for the two-armed bandit setting and subsequently extend the argument to the general multi-armed case. Without loss of generality, assume that the mean rewards satisfy \(\mu_1 > \mu_2\). Let \(\mathcal{M}\) denote an arbitrary class of reward distributions.
\begin{lemma}\label{lem:fanLB}
    Let $\pi$ be $\mathcal{M}^2$-optimal algorithm. Let $\gamma \in (0,1)$. Then, for any ${\boldsymbol{\nu}} \in \mathcal{M}^2$, for the suboptimal arm $2$,
    \begin{align*}
    \liminf_{T\to\infty}
    \frac{\log\mathbb{P}_{{\boldsymbol{\nu}}}^{\pi}(N_2(T)>(1-\gamma) T)}{\log T}
    \ge
    - \inf_{\substack{\widetilde{\nu}_1 \in \mathcal M: \\ \mathbb{E}[\widetilde{\nu}_1] \leq \mu_2 }} \frac{\mathrm{KL}(\widetilde\nu_1,\nu_1)}{\mathrm{KL}_{\inf}(\widetilde\nu_1,\mu_2)}.
    \end{align*}
\end{lemma}

\begin{proof}[Proof of Lemma~\ref{lem:fanLB}]   
Consider a two-armed stochastic bandit problem with an environment \({\boldsymbol{\nu}} = (\nu_1,\nu_2)\). Introduce an alternative environment \(\widetilde{\boldsymbol{\nu}} = (\widetilde{\nu}_{1}, \nu_{2}) \in \mathcal{M}^2\)  with means $(\widetilde{\mu}_1, \mu_2)$
that satisfy \(\widetilde{\mu}_1 < \mu_2\); hence, arm \(1\) becomes sub-optimal under \(\widetilde{\boldsymbol{\nu}}\). 
For a fixed \(\gamma \in (0,1)\) and  $T\ge 1$, define the event 
$$E = \{ N_2(T) > (1-\gamma)T \}.$$ 
Let $Y_a(s)$ denote the reward from $s^{\text{th}}$ pull of arm $a$.  Applying a change-of-measure argument from \({\boldsymbol{\nu}}\) to \(\widetilde{\boldsymbol{\nu}}\), we obtain the following: 
\begin{equation}\label{eq:likelihood}
\mathbb{P}_{{\boldsymbol{\nu}}}^{\pi}\bigl(E\bigr)
= \underset{\mathbb{P}_{\widetilde{\boldsymbol{\nu}}}^{\pi}}{\mathbb{E}} \lrp{\mathbb{I}(E)\frac{d\mathbb{P}_{{\boldsymbol{\nu}}}^{\pi}}{d\mathbb{P}_{\widetilde{\boldsymbol{\nu}}}^{\pi}}}
     = \int_{E} e^{L_T({\boldsymbol{\nu}},\widetilde{\boldsymbol{\nu}})} d\mathbb{P}^{\pi}_{\widetilde{\boldsymbol{\nu}}},
\end{equation}
where the log-likelihood‐ratio process $L_T({\boldsymbol{\nu}},\widetilde{\boldsymbol{\nu}})$ is as follows:
\[
L_T({\boldsymbol{\nu}},\widetilde{\boldsymbol{\nu}}):=
\log \bigg(\frac{d\mathbb{P}_{{\boldsymbol{\nu}}}^{\pi}}{d\mathbb{P}_{\widetilde{\boldsymbol{\nu}}}^{\pi}}\bigg)=
\log \bigg(\prod_{s=1}^{N_1(T)}
  \frac{d{\nu_1}}{d{\widetilde{\nu}_1}} \bigl(Y_1(s)\bigr)\bigg) = \sum_{s=1}^{N_1(T)}
    \log \biggl(\frac{d{\nu_1}}{d{\widetilde{\nu}_1}}\bigl(Y_1(s)\bigr)\biggr).
\]


Now we use the following result from \citet{FanGlynn2024Fragility}, which is stated below for completeness.
\begin{lemma}\label{lem:mp-optimal-pulls}
Let $\pi$ be an $\mathcal{M}^K$-optimal algorithm. Then, for any $\boldsymbol{\nu} \in \mathcal{M}^K$, and each sub-optimal arm $i$,
\begin{align*}
    \frac{N_i(T)}{\log T}
    \;\xrightarrow[T\to\infty]{\mathbb{P}_{\boldsymbol{\nu}}^{\pi}}\;
    \frac{1}{\mathrm{KL}_{\inf}\!\left(\nu_i,\mu_1\right)}.
\end{align*}
\end{lemma}
In $\widetilde{\boldsymbol{\nu}}$, since arm 1 is suboptimal, using Lemma~\ref{lem:mp-optimal-pulls},
\begin{equation}\label{eq-9}
\frac{N_1(T)}{\log T} \xrightarrow[T \to \infty]{\mathbb{P}_{\widetilde{\boldsymbol{\nu}}}^{\pi}}\; \frac{1}{\mathrm{KL}_{
\inf
}(\widetilde\nu_1,\mu_2)}.
\end{equation}

Further, under \(\widetilde{\boldsymbol{\nu}}\), by the weak law of large numbers as \(T\to\infty\),
\begin{equation}\label{eq-10}
\frac{1}{N_1(T)} \sum_{s=1}^{N_1(T)} 
    \log \biggl(\frac{d{\nu_1}}{d{\widetilde{\nu}_1}}\bigl(Y_1(s)\bigr)\biggr)
\;\xrightarrow{\mathbb{P}_{\widetilde{\boldsymbol{\nu}}}^{\pi}}\;
-\,\mathrm{KL}(\widetilde\nu_1,\nu_1).
\end{equation}
Combining Equations~\eqref{eq-9} and~\eqref{eq-10}, we get that under  $\widetilde{\boldsymbol{\nu}}$, 
\begin{equation*}
\frac{L_T({\boldsymbol{\nu}},\widetilde{\boldsymbol{\nu}})}{\log T} = \frac{N_1(T)}{\log T} \cdot \frac{1}{N_1(T)}\sum_{s=1}^{N_1(T)} 
    \log \biggl(\frac{d{\nu_1}}{d{\widetilde{\nu}_1}}\bigl(Y_1(s)\bigr)\biggr)
\;\xrightarrow[T \to \infty]{\mathbb{P}_{\widetilde{\boldsymbol{\nu}}}^{\pi}}\;
- \, \frac{\mathrm{KL}\left(\widetilde{\nu}_1, \nu_1\right)}{\mathrm{KL}_{\inf}\!\left(\widetilde{\nu}_1, \mu_2\right)} 
    .
\end{equation*}
From above, for any \(\varepsilon>0\) for large values of $T$, the following holds,
\begin{equation}
L_T({\boldsymbol{\nu}},\widetilde{\boldsymbol{\nu}})
\;=\;
\sum_{s=1}^{N_1(T)}
    \log \biggl(\frac{d{\nu_1}}{d{\widetilde\nu_1}}\bigl(Y_1(s)\bigr)\biggr)
\;\ge\;
-\,(1+\varepsilon) \frac{\mathrm{KL}(\widetilde\nu_1,\nu_1)}{\mathrm{KL}_{\inf}(\widetilde\nu_1,\mu_2)}\log T,
\end{equation}
and hence, for 
$$c_T = - (1+\varepsilon) \frac{\mathrm{KL}(\widetilde\nu_1,\nu_1)}{\mathrm{KL}_{\inf}(\widetilde\nu_1,\mu_2)}\log T, $$ 
$\mathbb{P}_{\widetilde{\boldsymbol{\nu}}}^{\pi}(L_T < c_T)\to 0$. Further, as the second arm is optimal in  $\widetilde{\boldsymbol{\nu}}$, $\mathbb{P}_{\widetilde{\boldsymbol{\nu}}}^{\pi}(E)\to 1$.

Finally, we also have
\begin{align*}
\mathbb{P}_{\widetilde{\boldsymbol{\nu}}}^{\pi}\bigl(E\bigr)
&= \int_{E \cap \{L_T < c_T\}} e^{-L_T({\boldsymbol{\nu}},\widetilde{\boldsymbol{\nu}})} d\mathbb{P}^\pi_{\boldsymbol{\nu}} + \int_{E \cap \{L_T \geq c_T\}} e^{-L_T({\boldsymbol{\nu}},\widetilde{\boldsymbol{\nu}})} d\mathbb{P}^\pi_{\boldsymbol{\nu}} \\
&\leq \int_{\{L_T < c_T\}} e^{-L_T(\boldsymbol{\nu},\widetilde{\boldsymbol{\nu}})} d\mathbb{P}^\pi_{{\boldsymbol{\nu}}} + \int_E e^{-c_T} d\mathbb{P}^\pi_{{\boldsymbol{\nu}}} \\
&\leq \mathbb{P}^\pi_{\widetilde{\boldsymbol{\nu}}}(\{L_T < c_T\}) + e^{-c_T} \mathbb{P}^\pi_{{\boldsymbol{\nu}}}(E),
\end{align*}
which implies
\[\mathbb{P}^\pi_{\boldsymbol{\nu}}(E) \geq e^{c_T} \big(\mathbb{P}^\pi_{\widetilde{\boldsymbol{\nu}}}(E) - \mathbb{P}^\pi_{\widetilde{\boldsymbol{\nu}}}(L_T < c_T)\big). \numberthis\label{eq-12}\]


Taking logarithm on both sides in Equation~\eqref{eq-12},  letting $\varepsilon \downarrow 0$, and optimizing over the free variable $\widetilde{\nu}_1 \in \mathcal M$,  we have, 
\begin{align*}
    \liminf_{T\to\infty}
\frac{\log\mathbb{P}_{\boldsymbol{\nu}}^{\pi}(N_2(T)>(1-\gamma)T)}{\log T}
\;\ge\; -
\inf_{\substack{\widetilde{\nu}_1 \in \mathcal M: \\ \mathbb{E}[\widetilde{\nu}_1] \leq \mu_2} } \frac{\mathrm{KL}(\widetilde\nu_1,\nu_1)}{\mathrm{KL}_{\inf}(\widetilde\nu_1,\mu_2)},
\end{align*}
proving the result.
\end{proof}

\noindent \textbf{Multi-arm setting:} Now we extend the above result to the setting with more than two arms. 
Without loss of generality, suppose that the means in the environment \(\boldsymbol{\nu} = (\nu_{1}, \nu_{2}, \ldots, \nu_{K})\) satisfy \(\mu_{1} > \mu_{2} > \cdots > \mu_{K}\).
Consider a new environment $\widetilde{\boldsymbol{\nu}}
= (\widetilde{\nu}_{1}, \widetilde{\nu}_{2}, \ldots, \widetilde{\nu}_{i-1}, \nu_{i}, \ldots, \nu_{K}) \in \mathcal{M}^{K}$ with respective means $(\widetilde{\mu}_1, \dots, \widetilde{\mu}_{i-1}, \mu_i, \dots, \mu_K)$ satisfying \(\widetilde{\mu}_{j} < \mu_{i}\) for all \(j<i\). Under this environment, arm \(i\) becomes the optimal arm. Similar to the above, $Y_a(s)$ denotes the reward when arm $a$ is played for $s$ instances. Let the event $E = \{ N_{i}(T) > (1-\gamma)T \}$ for some \(\gamma \in (0,1)\). Thus, by a change of measure from \(\boldsymbol{\nu}\) to \(\widetilde{\boldsymbol{\nu}}\), we have
\begin{equation*}
\mathbb{P}_{{\boldsymbol{\nu}}}^{\pi}\bigl(E\bigr)
= \underset{\mathbb{P}_{\widetilde{\boldsymbol{\nu}}}^{\pi}}{\mathbb{E}} \lrp{\mathbb{I}(E)\frac{d\mathbb{P}_{{\boldsymbol{\nu}}}^{\pi}}{d\mathbb{P}_{\widetilde{\boldsymbol{\nu}}}^{\pi}}}
     = \int_{E} e^{L_T({\boldsymbol{\nu}},\widetilde{\boldsymbol{\nu}})} d\mathbb{P}^{\pi}_{\widetilde{\boldsymbol{\nu}}},
\end{equation*}

where the log-likelihood‐ratio process $L_T(\boldsymbol{\nu},\widetilde{\boldsymbol{\nu}})$ is as follows:
\[
L_T(\boldsymbol{\nu}, \widetilde{\boldsymbol{\nu}})
\;:=\;
\log \bigg(\frac{d\mathbb{P}_{{\boldsymbol{\nu}}}^{\pi}}{d\mathbb{P}_{\widetilde{\boldsymbol{\nu}}}^{\pi}}\bigg)
\;=\;
\log \bigg( \prod_{j=1}^{i-1} \prod_{s=1}^{N_j(T)}
  \frac{d{\nu_j}}{d{\widetilde{\nu}_j}} \bigl(Y_j(s)\bigr)\bigg) = \sum_{j=1}^{i-1} L_j(\boldsymbol{\nu},\widetilde{\boldsymbol{\nu}}),
\]
where, for each arm $j \in \{1,2,\ldots,i-1\}$,  $L_j(\boldsymbol{\nu},\widetilde{\boldsymbol{\nu}})$ is as follows:
\[
L_j(\boldsymbol{\nu},\widetilde{\boldsymbol{\nu}})=\sum_{s=1}^{N_j(T)}
    \log\biggl(\frac{d{\nu_j}}{d{\widetilde{\nu}_j}}\bigl(Y_j(s)\bigr)\biggr).
\]
Proceeding along the same lines as in the two-armed bandit analysis, the preceding arguments can be extended to the general multi-armed setting, leading to the following for any suboptimal arm $i$: 
\begin{align}
    \liminf_{T\to\infty}
    \frac{\log\mathbb{P}_{\nu}^{\pi}(N_i(T) > (1-\gamma) T)}{\log T}
    \ge
    - \sum^{i-1}_{j=1} \inf_{\substack{\widetilde{\nu}_j \in \mathcal{M}: \\ \mathbb{E}[\widetilde{\nu}_j] \leq \mu_i }}\frac{\mathrm{KL}(\widetilde\nu_j,\nu_j)}{\mathrm{KL}_{\inf}(\widetilde\nu_j,\mu_i)} \label{eq:37}.
\end{align}
The desired result (extension of the above bound to deviations in \(\mathcal{D}_\gamma(T)=\bigl[\log^{1+\gamma}(T),\, (1-\gamma)T\bigr]\)) now follows from using Lemma~\ref{lem:ext1} (taken from \citet{FanGlynn2024Fragility}) in the above bound.  $\Box$


\begin{lemma}\label{lem:ext1}
Let $\boldsymbol{\nu}$ be any bandit environment with $B_{\gamma}(T) = [g(T), (1 - \gamma)T]$
and any strictly increasing function $g:(1, \infty) \to (0, \infty)$ such that  
$\underset{t \to \infty}{\lim} \tfrac{g(t)}{\log t} = \infty$
and $g(t) = o(t)$, and let $i$ be a sub-optimal arm in $\boldsymbol{\nu}$.  If the following condition holds  
\begin{equation*} \label{eq:50}
    \liminf_{T \to \infty} \frac{\log \mathbb{P}_{\boldsymbol{\nu}}^{\pi}\big(N_i(T) > (1-\gamma)T\big)}{\log T} \geq -c_i(\boldsymbol{\nu}),
\end{equation*}
then 
\begin{equation*} \label{eq:51}
    \liminf_{T \to \infty} \inf_{x \in B_{\gamma}(T)} \frac{\log \mathbb{P}_{\boldsymbol{\nu}}^{\pi}\big(N_i(T) > x\big)}{\log x} \geq -c_i(\boldsymbol{\nu}).
\end{equation*}  
\end{lemma}


\subsection{Proof of Theorem~\ref{theorem-3}}\label{appendix:proof_theo3}
\noindent \textbf{Two-arm setting:} To establish this theorem, we first prove the following lemma in a simplified two-armed bandit setting where arm \(1\) is optimal, i.e., \(\mu_1 > \mu_2\). 
As introduced earlier, let \(\mathcal{L}\) denote the class of reward distributions satisfying Assumptions~\ref{prop-1} and~\ref{prop-2}.

\begin{lemma}\label{lem-A.16}
    Let $\pi$ be the $\mathcal{L}^2$-optimal Generalized $\mathrm{KL}_{\inf}$-UCB algorithm. For $\gamma \in (0,1)$, let $x_T = \log^{1+\gamma}(T)$. Then, for ${\boldsymbol{\nu}} \in \mathcal{L}^K$, the suboptimal arm $2$ satisfies,
    \begin{equation}\label{eq:tail-upper-bound}
    \limsup_{T \to \infty}
    \frac{\log \mathbb{P}_{{\boldsymbol{\nu}}}^\pi\!\left(N_{2}(T) > x_T\right)}{\log x_T}
    \;\le\; - 1.
    \end{equation}
\end{lemma}

\begin{proof}[Proof of Lemma~\ref{lem-A.16}]\label{proof:lem-A.16} Consider a two-armed multi-armed bandit problem with environment 
\(\boldsymbol{\nu} = (\nu_1, \nu_2)\). Without loss of generality, assume \(\mu_1 > \mu_2\). For $m\ge 1$, let \(\tau_2(m)\) denote the time at which arm \(2\) is pulled for the \(m^{\text{th}}\) time. Fix  \(\delta \in (0,\, \mu_1 - \mu_2)\), and let \(C_{\nu} \geq 1\) be a constant satisfying \(x_T^{C_{\nu}} < T\). Then, 
\begin{align*}
\mathbb{P}_{\boldsymbol{\nu}}^{\pi}\left( N_2(T) > x_T \right) 
&\leq \mathbb{P}_{\boldsymbol{\nu}}^{\pi}( \exists\, t \in (\tau_2(x_T), T] \text{ s.t. } \{U_{1}(N_1(t-1), t-1) \leq U_{2}(N_2(t-1), t-1)\} \\
&\qquad \cap \{N_2(t-1) = x_T\})\\
&\leq \mathbb{P}_{\boldsymbol{\nu}}^{\pi}\left( \exists\, t \in (x_T, T] \text{ s.t. } U_{1}(N_1(t-1), t-1) \leq U_{2}(x_T, T) \right) \\
&\leq \mathbb{P}_{\boldsymbol{\nu}}^{\pi}\left( \exists\, t \in (x_T, T] \text{ s.t. } U_{1}(N_1(t-1), t-1) \leq \mu_2 + \delta \right) \numberthis\label{eq-43}\\
& \qquad+ \mathbb{P}_{\boldsymbol{\nu}}^{\pi}\left( U_{2}(x_T, T) > \mu_2 +\delta\right) \numberthis\label{eq-44} \\
&\leq \mathbb{P}_{\boldsymbol{\nu}}^{\pi}\left( \exists\, t \in (x_T, T] \text{ s.t. } U_{1}(N_1(t-1), t-1) \leq \mu_1 \right)  \\
&\qquad + \mathbb{P}_{\boldsymbol{\nu}}^{\pi}\left( U_{2}(x_T, T) > \mu_2 +\delta\right) \\
&= \underbrace{\mathbb{P}_{\boldsymbol{\nu}}^{\pi}\left( \exists\, t \in (x_T, x^{C_\nu}_T] \text{ s.t. } U_{1}(N_1(t-1), t-1) \leq \mu_1 \right)}_\text{A}\numberthis\label{eq-47} \\
& \quad + \underbrace{\mathbb{P}_{\boldsymbol{\nu}}^{\pi}\left( \exists\, t \in (x^{C_\nu}_T, T] \text{ s.t. } U_{1}(N_1(t-1), t-1) \leq \mu_1 \right)}_\text{B} \numberthis\label{eq:term-B}\\
&\qquad + \underbrace{\mathbb{P}_{\boldsymbol{\nu}}^{\pi}\left( U_{2}(x_T, T) > \mu_2 +\delta\right)}_\text{C}.
\end{align*}


\noindent  \textbf{Controlling A and B:} We upper bound A using Assumption~\ref{prop-1} with \(f_a(t) = \log(t) + 2\log\log(t) + g(N_a(t))\), \(g(N_a(t))=C_1\log(1+N_a(t)) + C_2\), and $C_1,C_2 >0$, as discussed next. 
\begin{align*}
\text{A} 
&= \mathbb{P}_{\boldsymbol{\nu}}^{\pi}\left( \exists\, t \in (x_T, x_T^{C_\nu}] \text{ s.t. } \mathrm{KL_{inf}}(\hat{\nu}_1(t-1),\mu_1) \geq \frac{f_1(t-1)}{N_1(t-1)}  \right) \\
&= \mathbb{P}_{\boldsymbol{\nu}}^{\pi} \{\exists\, t \in (x_T, x_T^{C_\nu}] \text{ s.t. } N_1(t-1)\mathrm{KL_{inf}}(\hat{\nu}_1(t-1),\mu_1) - g(N_1(t-1)) \notag\\
& \hspace{8cm} \geq \log(t-1)+\log(\log(t-1))\} \\
&\leq \mathbb{P}_{\boldsymbol{\nu}}^{\pi}\{\exists\, t \in (x_T, x_T^{C_\nu}] \text{ s.t. } N_1(t-1)\mathrm{KL_{inf}}(\hat{\nu}_1(t-1),\mu_1) - g(N_1(t-1))  \notag \\
& \hspace{8cm} \geq \log(x_T)+\log(\log(x_T))\} \\
&\leq \mathbb{P}_{\boldsymbol{\nu}}^{\pi}\{\exists\, t \in \mathbb{N} \text{ s.t. } N_1(t-1)\mathrm{KL_{inf}}(\hat{\nu}_1(t-1),\mu_1) - g(N_1(t-1)) \notag \\
& \hspace{8cm} \geq \log(x_T)+\log(\log(x_T))\} \\
&\leq \exp(-\log x_T-\log\log x_T)  \tag{Assumption~\ref{prop-1}}\\
&= \frac{1}{x_T\log x_T}.
\end{align*}
Term B can be upper-bounded in an analogous manner as below:
\[
\text{B} \leq {1/(x_T^{C_\nu}\log x_T^{C_\nu}}).
\]
Finally, summing the contributions from A and B, and taking the logarithm, we get
    \[\log (\text{A}+\text{B}) \leq - C_\nu \log x_T - \log \log x_T^{C_\nu} + \log (1+x_T^{C_\nu-1}C_\nu),\]
    which gives
    \[\frac{\log (\text{A}+\text{B})}{\log x_T } \leq -C_\nu - \frac{\log (C_\nu  \log x_T)}{\log x_T} + \frac{\log(1+x_T^{C_\nu-1}C_\nu)}{\log x_T},\]
    which further implies
    \[\limsup_{T \to \infty} \frac{\log (\text{A}+\text{B})}{\log x_T } \leq -C_\nu + (C_\nu-1) = -1.\]

\noindent \textbf{Controlling C:}
To upper bound C, we use 
Assumption~\ref{prop-2} with $d = \mathrm{KL}_{\inf}(\nu_{2},\, \mu_{2}+\delta) - \tfrac{f_2(T)}{x_{T}}$. Let $\hat{\nu}_{a,n}$ denote the empirical distribution of arm $a$ with $n$ samples. Note that $\tfrac{f_2(T)}{x_T}$ is decreasing in $T$. Thus, for all sufficiently large \(T\), we have  for constant \(c_\nu>0\) from Assumption~\ref{prop-2},
\begin{align*}
\text{C} &= \mathbb{P}_{\nu}^{\pi}\left(  U_{2}(x_T, T) > \mu_2 + \delta \right) \\
&= \mathbb{P}_{\boldsymbol{\nu}}^{\pi}\left( \mathrm{KL}_{\inf}(\hat{\nu}_{2,x_T}, \mu_2+\delta ) \leq \frac{f_2(T)}{x_T} \right) \\
&= \mathbb{P}_{\boldsymbol{\nu}}^{\pi}\left( \mathrm{KL}_{\inf}(\hat{\nu}_{2,x_T}, \mu_2+\delta) \leq \mathrm{KL}_{\inf}(\nu_2, \mu_2+\delta) - \left(\mathrm{KL}_{\inf}(\nu_2, \mu_2+\delta) - \frac{f_2(T)}{x_T} \right) \right) \\
&\leq \exp \left[- x_T c_\nu \left(\mathrm{KL}_{\inf}(\nu_2,\mu_2+\delta) - \frac{f_2(T)}{x_T}\right)^2\right].
\end{align*}


It can be easily shown that as $T \to \infty$,  $\frac{C}{A+B} \to 0$. Finally, summing all the terms, we get
\begin{align*}
\limsup_{T \to \infty}
    \frac{\log \mathbb{P}_{{\boldsymbol{\nu}}}^\pi\left(N_{2}(T) > x_T\right)}{\log x_T} &\le\limsup_{T\to\infty} \frac{\log(\text{A}+\text{B}+\text{C})}{\log x_T}  \\
    &= \limsup_{T\to\infty} \frac{\log (\text{A}+\text{B})}{\log x_T} + \limsup_{T\to\infty} \frac{\log\!\left(1+\frac{\text{C}}{\text{A+B}}\right)}{\log x_T} \\
    &\leq \limsup_{T\to\infty} \frac{\log (\text{A}+\text{B})}{\log x_T} + \limsup_{T\to\infty} \frac{\frac{\text{C}}{\text{A+B}}}{\log x_T} \\
    &\leq -1,
\end{align*}
proving the bound. 
\end{proof}

\noindent \textbf{Multi-arm setting:} We now extend the above argument to the setting with more than two arms. Without loss of generality, assume that $\mu_1 > \mu_2 > \dots > \mu_K$ in the environment \(\boldsymbol{\nu} = (\nu_1, \nu_2, \dots, \nu_K) \in \mathcal{L}^K\). Choose \(\delta \in (0, \min_{i\ge 3} \mu_{i-1} - \mu_i)\). Then, analogous to Equations~\eqref{eq-43} and~\eqref{eq-44}, we obtain the following decomposition for any suboptimal arm $i \ge 3$:
\begin{align*}
&\mathbb{P}_{\boldsymbol{\nu}}^{\pi}\left( N_i(T) > x_T \right) \\
&\leq \mathbb{P}_{\boldsymbol{\nu}}^{\pi}\left( \exists\, t \in (x_T, T] \text{ s.t. } \max_{1 \leq j \leq i-1}U_{j}(N_j(t-1), t-1) \leq \mu_i +\delta \right) \\
& \qquad + \mathbb{P}_{\boldsymbol{\nu}}^{\pi}\left(U_i(x_T,T) > \mu_i+\delta \right) \\
&\leq \mathbb{P}_{\boldsymbol{\nu}}^{\pi}\left(\forall 1 \leq j \leq i-1, \exists\, t \in (x_T, T] \text{ s.t. } U_{j}(N_j(t-1), t-1) \leq \mu_i +\delta \right) \\
& \qquad + \mathbb{P}_{\boldsymbol{\nu}}^{\pi}\left(U_i(x_T,T) > \mu_i+\delta \right) \\
&\leq \mathbb{P}_{\boldsymbol{\nu}}^{\pi}\left(\forall 1 \leq j \leq i-1, \exists\, t \in (x_T, x_T^{C_\nu}] \text{ s.t. } U_{j}(N_j(t-1), t-1) \leq \mu_j \right) \numberthis \label{eq:EA} \\
&\qquad + \mathbb{P}_{\boldsymbol{\nu}}^{\pi}\left(\forall 1 \leq j \leq i-1, \exists\, t \in (x_T^{C_\nu},T] \text{ s.t. } U_{j}(N_j(t-1), t-1) \leq \mu_j \right) \numberthis \label{eq:EB}\\
& \qquad + \mathbb{P}_{\boldsymbol{\nu}}^{\pi}\left(U_i(x_T,T) > \mu_i+\delta \right) \\
&= \mathbb{P}_{\boldsymbol{\nu}}^{\pi}\left(\bigcap^{i-1}_{j=1} E^\text{A}_j\right)  + \mathbb{P}_{\boldsymbol{\nu}}^{\pi}\left(\bigcap^{i-1}_{j=1} E^\text{B}_j \right) + \mathbb{P}_{\boldsymbol{\nu}}^{\pi}\left(U_i(x_T,T) > \mu_i+\delta \right),
\end{align*}
where 
\begin{align*}
& E^\text{A}_j = \{\exists t \in (x_T,x_T^{C_\nu}] \text{ s.t. } U_{j}(N_j(t-1), t-1) \leq \mu_j\},\\
& E^\text{B}_j = \{\exists  t \in (x_T^{C_\nu},T] \text{ s.t. } U_{j}(N_j(t-1), t-1) \leq \mu_j\}.
\end{align*}

Note that $E^\text{A}_j$'s and $E^\text{B}_j$'s are not independent for $j \in [i-1]$. However, we now show that each $E^A_j \subseteq F^A_j$ and $E^B_j \subseteq F^B_j $, where $F^A_j$ and $F^B_j$ are defined below. Recall, $\hat{\nu}_{a,n}$ denote the empirical distribution of arm $a$ with $n$ samples. 

\begin{align*}
    E^A_j &= \left\{ \exists t \in (x_T, x^{C_\nu}_T] \text{ s.t. } U_j(N_j(t-1), t-1) \le \mu_j \right\}\\
    &\subseteq \left\{ \exists t \in (x_T, x^{C_\nu}_T] \text{ s.t. } N_j(t-1) \KLinf(\hat{\nu}_{j}(t-1), \mu_j) \ge f_j(t-1) \right\}\\
    &=\bigg\{ \exists t \in (x_T, x^{C_\nu}_T] \text{ s.t. } N_j(t-1) \KLinf(\hat{\nu}_{j}(t-1), \mu_j) - g(N_j(t-1)) \geq \\
    &\qquad \hspace{6cm} \log(t-1) + \log\log(t-1)\bigg\}\\
    &\subseteq \left\{  \exists n \in \mathbb{N} \text{ s.t. } n \KLinf(\hat{\nu}_{j,n}, \mu_j) - g(n)  \ge  \log(x_T) + \log\log(x_T) \right\} := F^\text{A}_j
\end{align*}

Finally,
\begin{align*}
    \mathbb{P}_{\boldsymbol{\nu}}^{\pi}\left(\bigcap^{i-1}_{j=1}E^\text{A}_j\right) &\leq \prod^{i-1}_{j=1} \mathbb{P}_{\boldsymbol{\nu}}^{\pi}\bigg( \exists\, n \in \mathbb{N} \text{ s.t. } n \KLinf(\hat{\nu}_{j,n}, \mu_j) -g(n) \geq \log x_T + \log \log x_T  \bigg) \numberthis \label{eq:ind}
\end{align*}
\eqref{eq:ind} holds due to the independence of events $F^\text{A}_j$ for each $j \in [i-1]$. Similarly, we can define $F^{\text{B}}_j$ as follows,
\[
F^{\text{B}}_j = \left\{  \exists n \in \mathbb{N} \text{ s.t. } n \KLinf(\hat{\nu}_{j,n}, \mu_j) - g(n)  \ge  \log(x_T^{C_\nu}) + \log\log(x_T^{C_\nu}) \right\}
\]
Following the similar steps as above and using the independence of events $F^\text{B}_j$, we have,
\begin{align*}
    \mathbb{P}_{\boldsymbol{\nu}}^{\pi}\left(\bigcap^{i-1}_{j=1}E^\text{B}_j\right) &\leq \prod^{i-1}_{j=1} \mathbb{P}_{\boldsymbol{\nu}}^{\pi}\bigg( \exists\, n \in \mathbb{N} \text{ s.t. } n \KLinf(\hat{\nu}_{j,n}, \mu_j) -g(n) \geq \log (x_T^{C_\nu}) + \log \log (x_T^{C_\nu})  \bigg) \numberthis \label{eq:ind1}
\end{align*}
Notice that these terms in \eqref{eq:ind} and \eqref{eq:ind1} can further be bounded using similar bounds as in the terms $\text{A}$ and $\text{B}$ respectively for the two-arm setting above. Finally, taking log, the product terms in \eqref{eq:ind} and \eqref{eq:ind1} become a summation, giving the following bound: 
\begin{equation}\label{eq-75}
    \limsup_{T \to \infty}
    \frac{\log \mathbb{P}_{{\boldsymbol{\nu}}}^\pi\!\left(N_{i}(T) > x_T\right)}{\log x_T}
    \;\le\; - (i-1).
    \end{equation}
The desired result (extension of the above bound to deviations in \(\mathcal{D}_\gamma(T)=\bigl[\log^{1+\gamma}(T),\, (1-\gamma)T\bigr]\)) now follows from using Lemma~\ref{lem:A-18} (taken from \citet{FanGlynn2024Fragility}) in the above bound.  $\Box$ 


\begin{lemma}\label{lem:A-18}
Let $\boldsymbol{\nu}$ be any bandit environment with $B_{\gamma}(T) = [g(T), (1 - \gamma)T]$
and any strictly increasing function $g : (1, \infty) \to (0, \infty)$ such that  
$\underset{t \to \infty}{\lim} \frac{g(t)}{\log t} = \infty$
and $g(t) = o(t)$, and let $i$ be a sub-optimal arm in $\boldsymbol{\nu}$.  
If the following condition holds  
\begin{equation} 
    \liminf_{T \to \infty} \frac{\log P_{\boldsymbol{\nu}}^{\pi}\big(N_i(T) > g(T)\big)}{\log g(T)} \leq -c_i(\boldsymbol{\nu}),
\end{equation}
then
\begin{equation} 
    \liminf_{T \to \infty} \inf_{x \in B_{\gamma}(T)} \frac{\log P_{\boldsymbol{\nu}}^{\pi}\big(N_i(T) > x\big)}{\log x} \leq -c_i(\boldsymbol{\nu}).
\end{equation}  
\end{lemma}

\subsection{Proof of Theorem~\ref{theorem-4}}\label{appendix:proof_theo4}
First, recall that the finitely-supported distribution class \(\mathcal{L}_{[a,b],\mathcal{X}}\)  satisfies Assumption~\ref{prop-1} with $g(n) := \log(1+n) + 1$ (see, \citet{HondaBounded10} and \citet[Lemma F.1]{agrawal2021optimal}). Hence, the Generalized $\KLinf$-UCB algorithm is optimal for this class. In fact, it is known to be optimal with a smaller threshold of $f_a(t)=\log t + \log\log t$ for all $a\in [K]$. We refer the reader to \citet{cappe2013kullback} for a proof. We will use this smaller threshold in the rest of our proof below.

\noindent \textbf{Two-arm setting:} As before, we begin by establishing a simplified version of the upper bound in the two-armed bandit setting, which will subsequently be extended to the general multi-armed case. Without loss of generality, we assume for the remainder of this argument that \(\mu_1 > \mu_2\). 

\begin{lemma}\label{lem:sanov_lemma}
    Let $\pi$ be $\mathcal{L}^2_{[a,b],\mathcal{X}}$-optimal Generalized $\KLinf$-UCB algorithm with $f_a(t)=\log t + \log\log t$. For $\gamma \in (0,1)$, let $x_T = \log^{1+\gamma}(T)$. Then, for any ${\boldsymbol{\nu}} \in \mathcal{L}^2_{[a,b],\mathcal{X}}$, and the suboptimal arm $2$,
    \begin{align*}
    \limsup_{T\to\infty}
    \frac{\log \mathbb{P}_{{\boldsymbol{\nu}}}^\pi\!\left(N_{2}(T) > x_T\right)}{\log x_T}
    \le -
    \inf_{\substack{\widetilde{\nu}_1 \in \mathcal L_{[a,b],\mathcal{X}}:\\ \mathbb{E}[\widetilde{\nu}_1] \leq \mu_2 }} \frac{\mathrm{KL}(\widetilde\nu_1,\nu_1)}{\mathrm{KL}_{\inf}(\widetilde\nu_1,\mu_2)}
    \end{align*}
\end{lemma}

\begin{proof}[Proof of Lemma~\ref{lem:sanov_lemma}]
Consider a 2-arm multi-armed bandit problem with environment ${\boldsymbol{\nu}} = (\nu_1,\nu_2)$. We consider $\mu_1 > \mu_2$ without loss of generality. Let's take $\tau_2(m)$ to denote the time when arm 2 is played for the $m^{th}$ time. As the exploration function is not arm-dependent, we drop the subscript $a$ and simply denote it by $f(t)$ for the rest of the proof.  Now for $\delta \in (0,\mu_1-\mu_2)$,
\begin{align*}
\mathbb{P}_{\nu}^{\pi}\left( N_2(T) > x_T \right) 
&\leq \mathbb{P}_{\boldsymbol{\nu}}^{\pi}( \exists\, t \in (\tau_2(x_T), T] \text{ s.t. } \{U_{1}(N_1(t-1), t-1) \leq U_{2}(N_2(t-1), t-1)\} \\
&\qquad \cap \{N_2(t-1) = x_T\}) \\
&\leq \mathbb{P}_{\nu}^{\pi}\left( \exists\, t \in (x_T, T] \text{ s.t. } U_{1}(N_1(t-1), x_T) \leq U_{2}(x_T, T) \right) \\
&\leq \underbrace{\mathbb{P}_{\nu}^{\pi}\left( \exists\, t \in (x_T, T] \text{ s.t. } U_{1}(N_1(t-1), x_T) \leq \mu_2 + \delta \right)}_\text{A} \\
&\quad + \underbrace{\mathbb{P}_{\nu}^{\pi}\left( U_{2}(x_T, T) > \mu_2 +\delta\right)}_\text{B}.
\end{align*}

\noindent\textbf{Controlling A:}
For this, we use Sanov's theorem~\citep{dembo2010large,Csiszar2006Sanov}, which is stated below for completeness.

\begin{theorem}[Sanov's Theorem]
    Let $\mathcal{P}(\Sigma)$ denote the class of distributions over an underlying set $\Sigma$, and $\mathcal{T} \subset \mathcal{P}(\Sigma)$ be a subset of distribution with $\mathcal{T}^0$ and $\bar{\mathcal{T}}$ denoting the interior and closure of $\mathcal{T}$ respectively. Now $(X_n)_{n\in \mathbb{N}}$ be a sequence of i.i.d random variables from drawn from a distribution $\nu \in \mathcal{P}(\Sigma)$. The sequence of the empirical distributions $(\hat{\nu}_n)_{n \in \mathbb{N}}$ satisfy the large deviation principle with rate function $\mathrm{KL}(.,\nu)$ as follows.
    \begin{equation*}
        - \inf_{\nu' \in \mathcal{T}^0} \mathrm{KL}(\nu',\nu) \leq  \liminf_{n \rightarrow \infty} \frac{1}{n} \log \mathbb{P}(\hat{\nu}_n \in \mathcal{T}) \leq \limsup_{n \rightarrow \infty} \frac{1}{n} \log \mathbb{P}(\hat{\nu}_n \in \mathcal{T}) \leq - \inf_{\nu' \in \bar{\mathcal{T}}} \mathrm{KL}(\nu',\nu).
    \end{equation*}
\end{theorem}
The above theorem represents an asymptotic result. However, when $\Sigma$ is a finite set, we get an exact finite sample result with is given in the following equation:
\begin{align}\label{eq-94}
    \mathbb{P}(\hat{\nu}_n \in \mathcal{T}) \leq (n+1)^{|\Sigma|} \exp\{ - n\inf_{\nu' \in \bar{\mathcal{T}}} \mathrm{KL}(\nu',\nu) \}.
\end{align}
In order to apply Equation~\eqref{eq-94}, consider a set $\mathcal{V}$ such that, $\mathcal{V} = \{P \in \mathcal{L}_{[a,b],\mathcal{X}}: \KLinf(P,\mu_2 + \delta) \geq \frac{f(x_T)}{m}\}$. Recall, $\hat{\nu}_{a,n}$ denote the empirical distribution of arm $a$ with $n$ samples.  Thus, we can further bound the term $\text{A}$ as follows. 
\begin{align*}
\text{A} &= \mathbb{P}_{\nu}^{\pi}\left( \exists\, t \in (x_T, T] \text{ s.t. } U_{1}(N_1(t-1), x_T) \leq \mu_2 + \delta \right) \\
&\leq \mathbb{P}_{\nu}^{\pi}\left( \exists\, m \in \mathbb{N} \text{ s.t. } U_{1}(m, x_T) \leq \mu_2 + \delta \right) \\
&\leq \sum^\infty_{m=1} \mathbb{P}_{\nu}^{\pi}\left(  \mathrm{KL}_{\inf}(\hat{\nu}_{1,m},\mu_2 + \delta) \geq \frac{f(x_T)}{m} \right) \label{eq-98} \\
&=  \sum^\infty_{m=1} \mathbb{P}_{\nu}^{\pi}\left(\hat{\nu}_{1,m} \in \mathcal{V} \right) \\
& {\leq}  \sum^\infty_{m=1} \exp\left[- m \inf_{\nu' \in \bar{\mathcal{V}}} \mathrm{KL}(\nu', \nu_1) + o(m) \right], \tag{From~\eqref{eq-94}}
\end{align*} 
where $o(m)=\log(1+m)^s$. Note that $\mathcal{V}$ is closed in the topology of weak convergence. Further, $\mathrm{KL}(\cdot,\cdot)$ is lower-semicontinuous in the same topology. Hence, there exists $\nu^* \in \mathcal{V}$ such that $ \mathrm{KL}(\nu^*, \nu_1)=  \inf_{\nu' \in \bar{\mathcal{V}}} \mathrm{KL}(\nu', \nu_1)$. Now we define $s_T$ and $C_\nu$ as follows.
\[
s_T  = \frac{2f(x_T)}{\mathrm{KL}(\nu^*,\nu_1)}C_\nu \hspace{1em} \text{ such that } \hspace{1em} C_\nu = \underset{\substack{\widetilde{\nu}_1 \in  \mathcal L_{[a,b],\mathcal{X}}\\:\mathbb{E}[\widetilde{\nu}_1]\leq \mu_2+\delta}}{\inf} \frac{\mathrm{KL}(\widetilde{\nu}_1,\nu_1)}{\mathrm{KL}_{\inf}(\widetilde{\nu}_1,\mu_2+\delta)}.
\]
Notice that $C_\nu \geq 1$ because by definition $\KLinf(\widetilde{\nu}_1,\mu_2+\delta) \leq \mathrm{KL}(\widetilde{\nu}_1,\nu_1)$. Also, as $\nu^* \in \mathcal{V}$, by definition of $\mathcal{V}$, $\KLinf(\nu^*,\mu_2+\delta) \geq f(x_T)/m$. Thus for $m > s_T$, 
\[\frac{1}{2}\mathrm{KL}(\nu^*,\nu_1) \geq \frac{f(x_T)}{m}C_\nu.\]

Now the term $\text{A}$ can be upper-bounded as follows.
\begin{align*}
    \text{A}  &\leq \sum^\infty_{m=1} \exp\left[- m  \mathrm{KL}(\nu^*, \nu_1) + o(m)\right] \\
    &= \sum^{s_T}_{m=1} \exp\left[-m \mathrm{KL}(\nu^*,\nu_1)+ o(m)\right] + \sum^{\infty}_{m=s_T+1} \exp\left[-m \mathrm{KL}(\nu^*,\nu_1)+ o(m)\right] \\
    &= \sum^{s_T}_{m=1} \exp\left[-m \mathrm{KL}_{\inf}(\nu^*,\mu_2+\delta) \cdot \frac{\mathrm{KL}(\nu^*,\nu_1)}{\mathrm{KL}_{\inf}(\nu^*,\mu_2+\delta)}+ o(m)\right] \\
    &\hspace{6cm}+ \sum^{\infty}_{m=s_T+1} \exp\left[-m \mathrm{KL}(\nu^*,\nu_1)+ o(m)\right] \\
    &\leq \sum^{s_T}_{m=1} \exp\left[- f(x_T) \cdot \underset{\substack{\widetilde{\nu}_1 \in  \mathcal L_{[a,b],\mathcal{X}}\\:\mathbb{E}[\widetilde{\nu}_1]\leq \mu_2+\delta}}{\inf} \frac{\mathrm{KL}(\widetilde{\nu}_1,\nu_1)}{\mathrm{KL}_{\inf}(\widetilde{\nu}_1,\mu_2+\delta)} + o(m)\right] \\
    &\hspace{6cm}+ \sum^{\infty}_{m=s_T+1} \exp\left[-m \mathrm{KL}(\nu^*,\nu_1)+ o(m)\right] \\
    &\leq \sum^{s_T}_{m=1} \exp\left[- f(x_T) \cdot C_\nu + o(m)\right] + \sum^{\infty}_{m=s_T+1} \exp\left[-m \mathrm{KL}(\nu^*,\nu_1) + o(m)\right] \\
    &\leq \sum^{s_T}_{m=1} \exp\left[- f(x_T) \cdot C_\nu \right] \exp[ o(m)] + \sum^{\infty}_{m=s_T+1} \exp\left[-f(x_T) C_\nu - \frac{m}{2} \cdot \mathrm{KL}(\nu^*,\nu_1) + o(m)\right] \\
    &= \exp\left[- f(x_T) \cdot C_\nu \right] \left( \sum^{s_T}_{m=1} \exp[o(m)] + \sum^{\infty}_{m=s_T+1} \exp\left[-  \frac{m}{2} \cdot \mathrm{KL}(\nu^*,\nu_1) + o(m) \right] \right).
\end{align*}

\noindent\ \textbf{Controlling the second term B:}
The term $\text{B}$ can be upper bounded using Assumption~\ref{prop-2} with  $d= \mathrm{KL}_{\inf}(\nu_2,\mu_2+\delta) - \frac{f(T)}{x_T}$.  Note that $\frac{f(T)}{x_T} $ decreases as $T$ increases. Thus, for all sufficiently large \(T\), we have \(d>0\) and for any constant \(c_\nu>0\),
\begin{align*}
\text{B}  &= \mathbb{P}_{\nu}^{\pi}\left(  U_{2}(x_T, T) > \mu_2 + \delta \right) \\
&= \mathbb{P}_{\nu}^{\pi}\left( \mathrm{KL}_{\inf}(\hat{\nu}_{2,x_T}, \mu_2+\delta ) \leq \frac{f(T)}{x_T} \right) \\
&= \mathbb{P}_{\nu}^{\pi}\left( \mathrm{KL}_{\inf}(\hat{\nu}_{2,x_T}, \mu_2+\delta) \leq \mathrm{KL}_{\inf}(\nu_2, \mu_2+\delta) - \left(\mathrm{KL}_{\inf}(\nu_2, \mu_2+\delta) - \frac{f(T)}{x_T} \right) \right) \\
&\leq \exp \left[- x_T c_\nu \left(\mathrm{KL}_{\inf}(\nu_2,\mu_2+\delta) - \frac{f(T)}{x_T}\right)^2\right].
\end{align*}

Finally, summing the contributions from the terms \(\text{A}\) and \(\text{B}\) and taking the logarithm, we obtain the following bound:

\begin{align*}
    \log(\text{A} + \text{B}) &\leq 
    - f(x_T) \cdot C_\nu \\
    &\quad +
    \log \Bigg(\underbrace{
     \sum^{s_T}_{m=1} \exp[o(m)] + \sum^{\infty}_{m=s_T+1} \exp\left[-  \tfrac{m}{2} \cdot \mathrm{KL}(\nu^*,\nu) + o(m) \right]} \\
    &\qquad \qquad + \underbrace{\exp \left[- x_T c_\nu 
    \left(\mathrm{KL}_{\inf}(\nu_2,\mu_2+\delta) - \tfrac{f(T)}{x_T}\right)^2 + f(x_T) C_\nu \right]}_\text{G}\Bigg). \numberthis \label{eq:sanov_lo}
\end{align*}

For distributions in $\mathcal L_{[a,b],\mathcal X}$, $o(m)=\log(m+1)^s$. Also recall that $f(T) = O(\log T)$ and $x_T = O(\log^{1+\gamma}(T))$. Using these we have,

\begin{align*}
    \frac{\log \text{G}}{f(x_T)} &= \frac{1}{f(x_T)} \log \Bigg( \sum^{s_T}_{m=1} \exp[o(m)] + \sum^{\infty}_{m=s_T+1} \exp\left[-  \tfrac{m}{2} \cdot \mathrm{KL}(\nu^*,\nu) + o (m)\right] \\
    &\hspace{2cm} 
    + \exp \left[-x_T \cdot c_\nu\left(\mathrm{KL}_{\inf}( \nu_2,\mu_1) - \frac{f(T)}{x_T}\right)^2 + C_\nu f(x_T)\right] \Bigg) \\
    &= \frac{1}{f(x_T)} \log \Bigg( \sum^{s_T}_{m=1} (m+1)^s + \sum^{\infty}_{m=s_T+1} 
    \exp\left[-  \tfrac{m}{2} \cdot \mathrm{KL}(\nu^*,\nu) + s\log(m+1)\right] \\
    &\hspace{2cm} 
    + \exp \left[-x_T \cdot c_\nu\left(\mathrm{KL}_{\inf}( \nu_2,\mu_1) - \frac{f(T)}{x_T}\right)^2 + C_\nu f(x_T)\right] \nonumber \Bigg) \\
    &\leq \frac{1}{f(x_T)} \log \Bigg( s_T (s_T+1)^s + O(\exp(-(s_T+1))) \\
    &\hspace{2cm} 
    + \exp \left[-x_T \cdot c_\nu\left(\mathrm{KL}_{\inf}( \nu_2,\mu_1) - \frac{f(T)}{x_T}\right)^2 + C_\nu f(x_T)\right] \nonumber \Bigg) \\
    &\leq \frac{1}{f(x_T)} \log \Bigg(O(f^{s+1}(x_T)) + O(\exp(-f(x_T))) \\
    &\quad
    +\exp \left(-x_T \cdot c_\nu \mathrm{KL}^2_{\inf}( \nu_2,\mu_1) + 2 c_\nu \mathrm{KL}_{\inf}(\nu_2,\mu_1)f(T) - c_\nu f^{1-\gamma}(T) + C_\nu f(x_T)\right) \Bigg) \numberthis\label{eq-24}.  \\
    \end{align*}

Note that as ${T \to \infty}$ the second and third term inside the log in Equation~\eqref{eq-24} go to zero and finally we get \[\limsup_{T\to \infty} \frac{\log G}{f(x_T)} \to \frac{\log O(f^{s+1}(x_T))}{f(x_T)}  \to 0.\]

Finally using this in Equation~\eqref{eq:sanov_lo}, taking $\delta \downarrow 0$, we get the desired result as follows.
\[
\limsup_{T\to\infty}
    \frac{\log \mathbb{P}_{{\boldsymbol{\nu}}}^\pi\!\left(N_{2}(T) > x_T\right)}{\log x_T}
    \;\le\; -
    \inf_{\substack{\widetilde{\nu}_1 \in \mathcal{L}_{[a,b],\mathcal{X}} \\ : \mathbb{E}[\widetilde{\nu}_1] \leq \mu_2 }} \frac{\mathrm{KL}(\widetilde\nu_1,\nu_1)}{\mathrm{KL}_{\inf}(\widetilde\nu_1,\mu_2)},
\]
completing the proof.
\end{proof}
\noindent \textbf{Multi-arm Setting:} The extension to the multi-arm setting follows along lines analogous to the same extension in the proof of Theorem~\ref{theorem-3} in Section~\ref{appendix:proof_theo3}. Finally, invoking Lemma~\ref{lem:A-18}, we obtain the following result for the deviation family  
\(\mathcal{D}_\gamma(T)=\bigl[\log^{1+\gamma}(T),\, (1-\gamma)T \bigr]\) with \(\gamma \in (0,1)\), as stated below.

 \begin{equation*}\label{eq:finite-model-result}
        \limsup_{T \to \infty}
    \inf_{x \in \mathcal{D}_\gamma(T)}
    \frac{\log \mathbb{P}_{{\boldsymbol{\nu}}}^\pi\!\left(N_{i}(T) > x\right)}{\log x}
    \leq - \sum_{j=1}^{i-1}
    \inf_{\substack{\widetilde{\nu} \in \mathcal{L}_{[a,b],\mathcal{X}}:\\ m(\widetilde{\nu}) < \mu_i}}
    \frac{\mathrm{KL}\!\left(\widetilde{\nu}, \nu_j\right)}
    {\mathrm{KL}_{\inf}\!\left(\widetilde{\nu}, \mu_i\right)}.
\end{equation*}
Now using the optimal regret tail lower bound from Equation~\eqref{eq:tail-lower-bound} we finally get the tight characterization as follows.
 \begin{equation*}\label{eq:finite-model-result}
        \lim_{T \to \infty}
    \inf_{x \in \mathcal{D}_\gamma(T)}
    \frac{\log \mathbb{P}_{{\boldsymbol{\nu}}}^\pi\!\left(N_{i}(T) > x\right)}{\log x}
    = - \sum_{j=1}^{i-1}
    \inf_{\substack{\widetilde{\nu} \in \mathcal{L}_{[a,b],\mathcal{X}}:\\ m(\widetilde{\nu}) < \mu_i}}
    \frac{\mathrm{KL}\!\left(\widetilde{\nu}, \nu_j\right)}
    {\mathrm{KL}_{\inf}\!\left(\widetilde{\nu}, \mu_i\right)}.
\end{equation*}

\section{Relevant literature}\label{sec:lit}

Expected regret minimization in stochastic multi-armed bandits has a long and rich history, dating back to early work by~\citet{Thomp1933} in clinical trials and \citet{robbins1952}. A foundational result was established by~\citet{LaiRobbins85}, who derived a lower bound on the expected cumulative regret for parametric reward models. This result shows that the expected regret must grow at least logarithmically in the time horizon $T$, with a leading constant determined by Kullback--Leibler (KL) divergences between the optimal and suboptimal arms. This was later generalized by~\citet{BURNETAS1996}, leading to the notion of asymptotically optimal algorithms beyond SPEFs.

Several algorithms are now known to achieve this lower bound. Among these, the KL-UCB algorithm of~\citet{cappe2013kullback} is a prominent example. While UCB-type algorithms have long been studied~\citep{Lai1987Adaptive,agrawal1995sample,Auer2002}, \citet{cappe2013kullback} provided a finite-time analysis of KL-UCB and showed its asymptotic optimality for SPEF, as well as for finitely-support models via empirical variants of $\KLinf$, the $\operatorname{KL}$-projection function. They also conjectured optimality for general bounded rewards, which was later established by~\citet{agrawal21heavytailed}. Since then, much of the literature has focused on designing different asymptotically optimal algorithms for various distribution families~\citep{HondaBounded10,honda2011asymptotically,SAgrawal2011TS1,SAgrawal2012TS2,kaufmann2012TS,cappe2013kullback,bubeck2013heavytail,honda2015SemiBounded,agrawal21heavytailed}. We refer the reader to~\citet{bubeck2012Survey,lattimore2020bandit} for an overview.

Beyond expected regret, there has been sustained interest in understanding finer properties of the regret distribution. Early works~\citep{Audibert2009UCBV,salomon2011deviations} studied deviation and concentration properties of regret. More recently, a line of work has focused on the distributional behavior of regret, including worst-case guarantees~\citep{KalvitZeevi2021WorstCaseBandits}. In parallel, \citet{FanGlynn2022TypicalBehavior} analyzed the typical behavior of regret, establishing strong laws and central limit theorems for bandit algorithms.

A complementary perspective, focusing on rare but significant deviations, was developed by~\citet{FanGlynn2024Fragility}. They show that even asymptotically optimal algorithms can exhibit heavy regret tails, leading to large regret with non-negligible probability. This phenomenon highlights a fundamental limitation of expectation-based optimality and motivates the study of regret tail behavior.
\end{document}